\documentclass[12pt]{article}
\usepackage{amssymb,amsmath,amsthm,array,graphicx,color}
\usepackage{hyperref}
\usepackage[a4paper,margin=3cm]{geometry}

\makeatletter

\@addtoreset{equation}{section}
\makeatother

\newtheorem{mydef}{Definition}
\newtheorem{prop}{Proposition}

\newcommand{\dblref}[2]{(\ref{#1}--\ref{#2})}
\newcommand{\figref}[1]{Fig.~\ref{fig:#1}}
\newcommand{\secref}[1]{Sec.~\ref{sec:#1}}
\newcommand{\appref}[1]{Appendix~\ref{sec:#1}}
\newcommand{\propref}[1]{Prop.~\ref{prop:#1}}

\newcommand{\aver}[1]{\langle {#1} \rangle}
\newcommand{\Aver}[1]{\langle\!\langle {#1} \rangle\!\rangle}

\newcommand{\ket}[1]{| {#1} \rangle}
\newcommand{\Ket}[1]{| {#1} \rangle\!\rangle}
\newcommand{\dblind}[2]{\begin{array}{c} {\scriptstyle #1} \\ {\scriptstyle #2} \end{array}}
\newcommand{\boxmath}[1]{\text{\fbox{$\displaystyle #1$}}}

\newcommand{\bino}[2]{\left(\begin{array}{c} #1 \\ #2 \end{array}\right)}

\newcommand{\ab}{{\bar a}}
\newcommand{\bb}{{\bar b}}
\newcommand{\db}{{\bar d}}
\newcommand{\alphab}{{\bar\alpha}}
\newcommand{\betab}{{\bar\beta}}
\newcommand{\C}{\mathcal{C}}
\newcommand{\ellb}{{\bar\ell}}
\newcommand{\F}{\mathcal{F}}
\newcommand{\Ft}{{\wt{\mathcal{F}}}}
\newcommand{\half}{\frac{1}{2}}
\newcommand{\hb}{{\bar h}}
\newcommand{\id}{{\bf 1}}
\renewcommand{\Im}{{\rm Im}\,}

\newcommand{\kb}{{\bar k}}
\newcommand{\Lb}{{\bar L}}
\newcommand{\Nb}{\mathbb{N}}
\newcommand{\ncup}{\ _\bullet\!\!\!\cup}
\newcommand{\nn}{\nonumber}
\newcommand{\nt}{\wt{n}}

\newcommand{\ol}{\overline}
\newcommand{\On}{O($n$)}
\newcommand{\PTL}{\mathrm{PTL}}

\newcommand{\V}{\mathcal{V}}
\newcommand{\Vir}{\mathrm{Vir}}
\newcommand{\vphi}{\varphi}
\newcommand{\vphib}{{\bar\varphi}}
\newcommand{\W}{\mathcal{W}}
\newcommand{\wh}{\widehat}
\newcommand{\wt}{\widetilde}
\newcommand{\Yb}{\Upsilon_\beta}
\newcommand{\zb}{{\bar z}}
\newcommand{\Zb}{\mathbb{Z}}
\newcommand{\Ie}{I_{1,\epsilon}}
\newcommand{\Ipe}{I_{2,\epsilon}}

\title{Correlation functions in loop models}
\author{B. Estienne and Y. Ikhlef}

\begin{document}
\maketitle

\begin{abstract}

  In this paper we provide a step towards the understanding of the {\On} bulk operator algebra. By using a mixture of analytical and numerical methods, we compute (ratios of) structure constants, and analyse the logarithmic structure of the transfer matrix. We believe that the {\On} model for a generic value of $n = q + q^{-1}$ (\emph{i.e.} for $q$ not a root of unity) provides a toy model of a bulk logarithmic CFT that is considerably simpler than its counterparts at $q$ a root of unity.

\end{abstract}

\tableofcontents

\section{Introduction}
\label{sec:intro}

The most striking phenomenon in statistical physics is maybe universality in the vicinity of a critical point. As a statistical system undergoes a second order phase transition, the correlation length diverges, various quantities become related through scaling relations, and many power-law divergences occur, governed by universal critical exponents. While this phenomenon is well explained in all dimensions by the ideas underlying the renormalisation group approach, it is extremely difficult to compute the universal properties in general.  At the critical point the system becomes scale invariant, which is usually promoted to conformal invariance, and the system exhibits fractal behaviour. The corresponding field theory is called a conformal field theory (CFT).

Unlike in higher dimensions, the bidimensional conformal symmetry has infinitely many generators. This makes the CFT approach extremely powerful, and many critical exponents have been calculated exactly. This is the case for loop models such as the {\On} loop model and the $Q$-state Potts model, where the full spectrum has been known for some time \cite{DSZ87}, including the critical exponents for geometric objects such as the fractal dimension of interfaces. These models are particularly interesting as they contain many primary  fields which are non-scalar (\emph{i.e.} fields with a non zero conformal spin). Both the {\On} model and the $Q$-state Potts model can be written as (non-local) loop models, which are well known to exhibit logarithmic behaviour. For this reason, the underlying CFTs are poorly understood beyond their spectrum.

A crucial part of the CFT data is given by its Operator Product Expansions (OPE), which are characterized by two main ingredients. The first one is the set of fusion rules, and the second one is the set of all structure constants, \emph{i.e.} the three point functions between primary fields. Together these define the operator algebra, which has to be associative in a consistent field theory. For the {\On} model and $Q$-state Potts model the operator algebra is mostly unknown, which means that we do not have a complete understanding of how to compute correlation functions. Only a handful of structure constants have been obtained through an analytic continuation of the Liouville theory~\cite{DV11,PSVD13}, but this method does not work for fields with a non-zero magnetic charge: a simple example where it fails is the probability for three given points to lie on the same loop.

Our objective is to describe the full operator algebra of these models, and in this paper we present a step in that direction. We focus on the Temperley-Lieb (TL) model, which is the loop model description of the $Q$-state Potts model. Moreover we restrict ourselves to generic values of the loop weight $n$  (\emph{i.e.} $n= q+ q^{-1}$ with $q$ not a root a unity).

The first tool we use is the Coulomb Gas description of the TL model. We extend the analysis of crossing symmetry of four-point functions to the case of non-scalar primary fields. While a subset of  the structure constants can be extracted this way, charge neutrality proves to be a strong limitation on the set of structure constants for non-scalar fields that can be accessed with this approach. Therefore we turn to a more powerful approach inspired by the bootstrap of Liouville theory~\cite{Teschner95}. It is a well known fact of CFT that the decoupling of a null vector yields a differential equation for correlation functions. In conjunction with crossing symmetry, it is possible to compute ratios of structure constants. Using an extension of the bootstrap method to the case of non-scalar fields, we make predictions for many ratios of structure constants. An interesting remark is the appearance of logarithms in some four-point functions, as expected from such non-diagonalisable CFTs. We also compute numerically these structure constants on the lattice, and we find an excellent agreement with our analytical calculation.

The logarithmic nature of the CFT raises a fundamental question regarding the bootstrap approach : do null-vectors decouple in the scaling theory of the {\On} model? Indeed, while in a unitary CFT the null vectors --  \emph{i.e.} states with a vanishing norm -- always decouple, logarithmic CFTs~\cite{Saleur87,RozSaleur91,Gurarie93,Pearce06} are known to have a non-definite inner product, allowing states of zero norm to have a non-vanishing inner product with other states. This is typically what happens in non-trivial Jordan cells of a logarithmic CFT. These questions can be also be addressed on the lattice: for the TL model the Jordan cells can already be observed in finite size on the lattice model (see \cite{Dubail10,Vasseur12,Azat12,Azat13,Azat14}), and this provides a powerful tool to analyse the decoupling of null-vectors.

The paper is organised as follows. In~\secref{model}, we recall the definition of lattice loop models, and review the operator content of the corresponding compact boson CFT (with a background charge) in the scaling limit~\cite{Nienhuis82,DSZ87}. In~\secref{boot}, we generalise the Coulomb-Gas~\cite{DF84,DF85} and conformal bootstrap~\cite{Teschner95,ZZ96} approaches to determine a family of OPE coefficients in this CFT. In~\secref{LCFT}, we come back to the lattice to study the exact Jordan blocks appearing in the spectrum of the transfer matrix (or Hamiltonian) of the loop models, and then discuss the associated structures in the CFT. Moreover, we give an example of analytic calculation of the indecomposability parameter $\beta$ in one of these Jordan blocks. Finally, we present some open questions and perspectives in~\secref{concl}. In the Appendix, some technical details of our analytical calculations are given, and our numerical method to compute OPE coefficients is outlined.

\section{The {\On} and Temperley-Lieb loop models}
\label{sec:model}

\subsection{Lattice models}

The {\On} loop model on the honeycomb lattice~\cite{Nienhuis82} consists in configurations of non-intersecting closed polygons. A configuration $C$ is given the Boltzmann weight
\begin{equation} \label{eq:On-boltzmann}
  W_{{\rm O}(n)}(C) = K^{\ell(C)} \ n^{N_{\rm c}(C)} \ \nt^{N_{\rm nc}(C)} \,,
\end{equation}
where $\ell(C)$ is the total length of polygons in $C$, $N_{\rm c}(C)$ is the number of contractible loops in $C$, and $N_{\rm nc}(C)$ is the number of non-contractible loops in $C$ (see \figref{On-config}).

For fixed $n,\nt$ both in the interval $[-2,2]$, the {\On} loop model has two distinct critical regimes: (i) when $K$ is at the critical value $K_c=1/\sqrt{2+\sqrt{2-n}}$, for example critical polymers ($n=0$), critical Ising domain walls ($n=1$), XY model ($n=2$); (ii) when $K>K_c$, the model is in the universality class of the $Q$-state Potts model, with $n=\sqrt{Q}$. In this paper, we will mainly focus on the universality class ii, since it has a simpler realisation, namely the Temperley-Lieb model.

The Temperley-Lieb model is defined on the square lattice, and consists in completely packed configurations of non-intersecting closed polygons, as shown in~\figref{TL-config}. A configuration $C$ is given the Boltzmann weight
\begin{equation} \label{eq:TL-boltzmann}
  W_{\rm TL}(C) = n^{N_{\rm c}(C)} \ \nt^{N_{\rm nc}(C)} \,,
\end{equation}
with the same notations as above. For completeness, we recall the relation~\cite{Baxter-book} between the Potts and TL models, in the case of a planar domain, and for $n=\nt$. Consider the $Q$-state Potts model on the lattice $\cal L$, with spin variables $\{s_j\}$ and Boltzmann weights $\exp(J \sum_{\aver{ij}} \delta_{s_i,s_j})$. The partition function can be graphically expanded using the identity $\exp(J \delta_{s_i,s_j}) = 1 + v \delta_{s_i,s_j}$, where $v = e^J-1$. This gives the Fortuin-Kasteleyn partition function
\begin{equation} \label{Z-FK}
  Z_{\rm Potts} = \sum_{\rm clusters} Q^{\# {\rm connected\ components}} \ v^{\# {\rm occupied\ edges}} \,,
\end{equation}
where the sum is over all possible subgraphs (or ``cluster configurations'') of the lattice $\cal L$.
Each cluster configuration can in turn be associated uniquely to a loop configuration of the TL model on the medial lattice~$\cal M$, i.e. the lattice connecting the midpoints of adjacent edges of~$\cal L$: see \figref{TL-config}.
Using the Euler relation, one gets
\begin{equation}
  Z_{\rm Potts} = Q^{N_s(\mathcal{L})} \times \sum_{{\rm loop \ config.} \ C} (v/\sqrt{Q})^{N_e(C)} \sqrt{Q}^{N_\ell(C)} \,,
\end{equation}
where $N_s(\mathcal{L})$ is the number of sites of $\cal L$, $N_e(C)$ is the number of edges of $\cal L$ not crossed by the loops, and $N_\ell(C)$ is the number of closed loops.

\begin{figure}
  \begin{center}
    \includegraphics[scale=0.5]{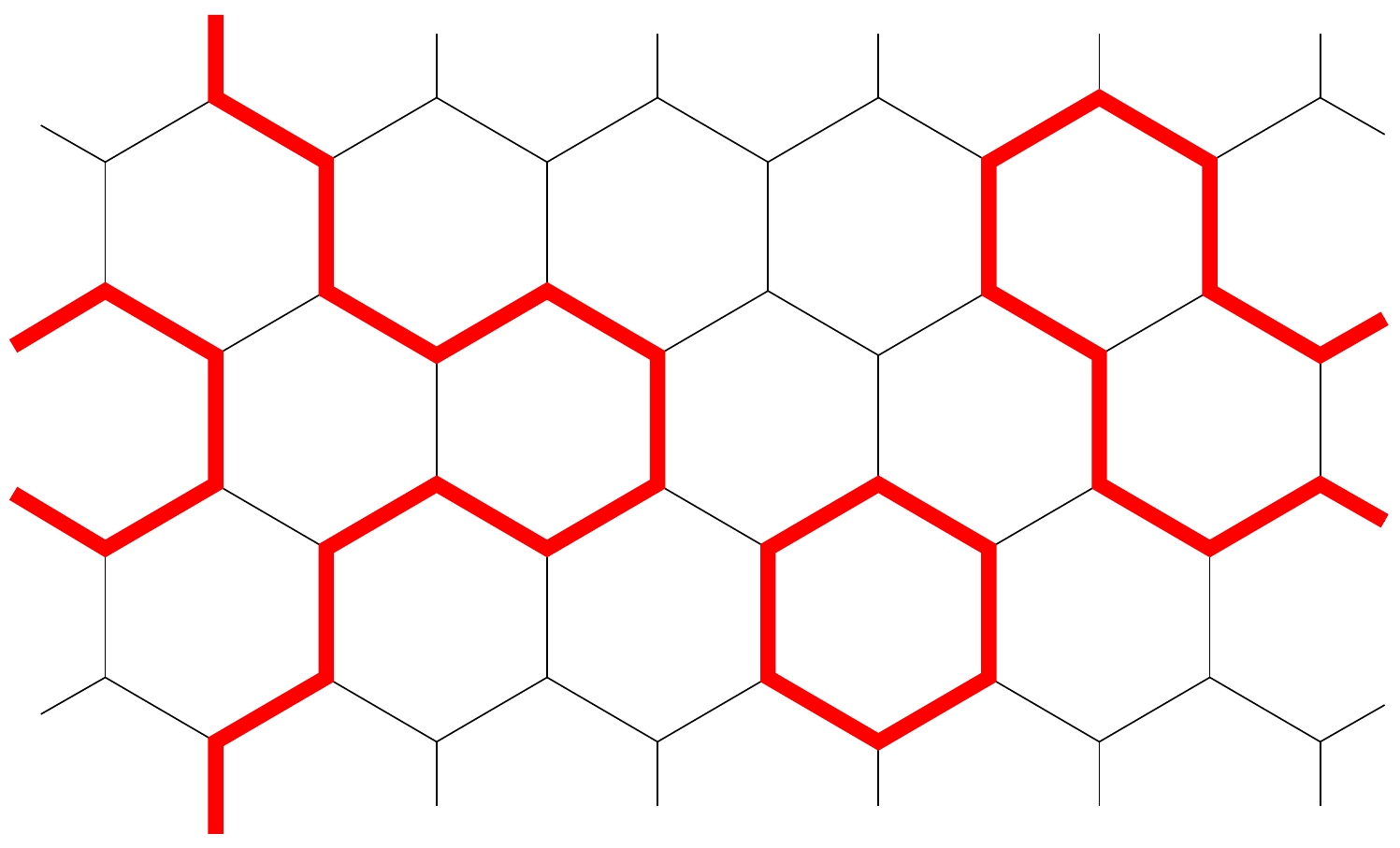}
  \end{center}
  \caption{An example configuration of the {\On} model on the honeycomb lattice
    embedded on a torus.}
  \label{fig:On-config}
\end{figure}

\begin{figure}
  \begin{center}
    \begin{tabular}{ccc}
      \includegraphics[scale=0.5]{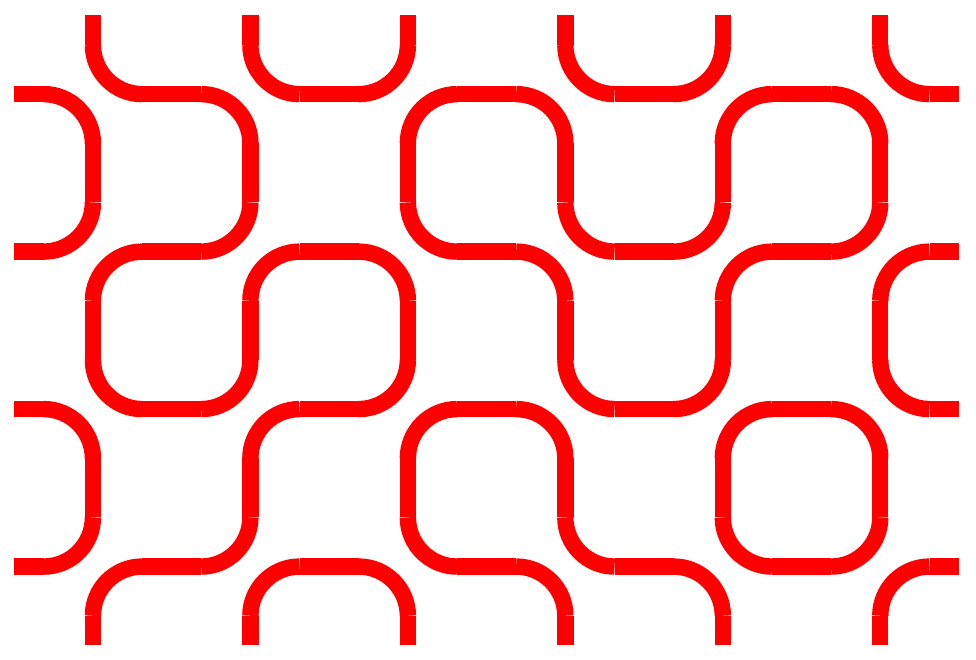}
      &&
      \includegraphics[scale=0.5]{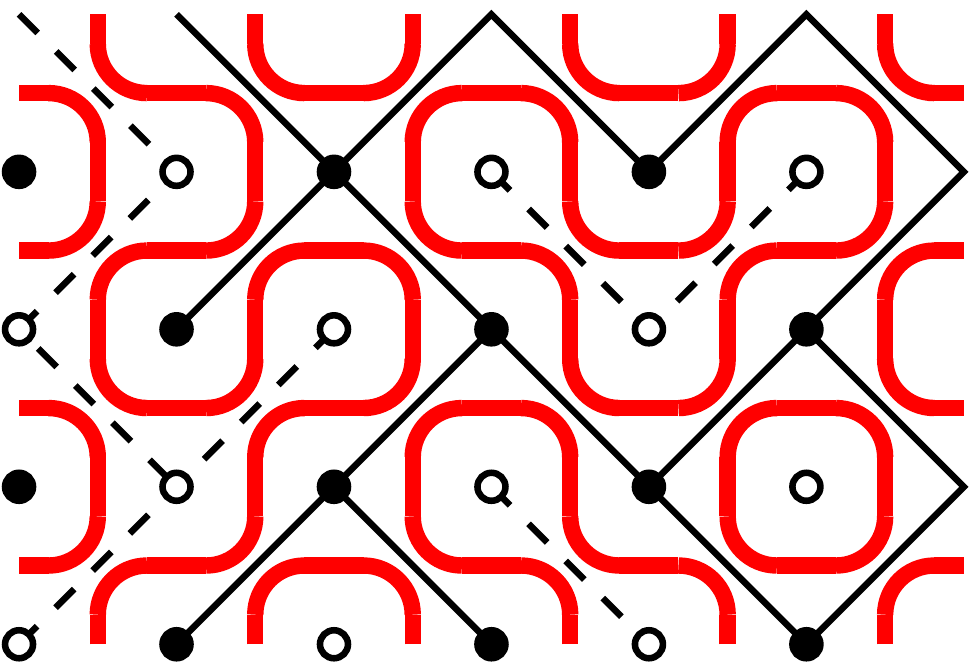} \\
      \\
      (a) &\qquad\qquad& (b)
    \end{tabular}
  \end{center}
  \caption{(a) An example configuration of the Temperley-Lieb model on the torus. (b) The same TL configuration, with the corresponding Fortuin-Kasteleyn clusters and dual clusters (the Potts spins live on the full dots).}
  \label{fig:TL-config}
\end{figure}

\subsection{Scaling limit: compact boson CFT}

The emergence of a compact boson in the scaling limit is best described for the honeycomb {\On} model on a torus. Take a configuration $C$ of the {\On} model, and give an orientation to each loop, independently of the others. Then, define discrete height variables $\{\phi_j\}$ at the centers of faces of the honeycomb lattice, so that the oriented loops are their contour lines, with a step $\pm \pi$ across each contour line: see \figref{height}.

Note that, for a contractible loop oriented anti-clockwise (resp. clockwise), the algebraic number of turns to the left is $+6$ (resp. $-6$). Hence, the factor $n^{N_{\rm c}(C)}$ in~\eqref{eq:TL-boltzmann} can be distributed locally by assigning each left (resp. right) turn a phase factor $e^{i\chi}$ (resp. $e^{-i\chi}$), with $n=2\cos 6\chi$.

In the scaling limit, it can be argued that the coarse-grained analog of $\phi_j$ is a free scalar field $\phi(\vec r)$. Moreover, the free scalar field $\phi(\vec r)$ has to be compact. Indeed, a given cycle of the torus may be crossed by a non-zero flux $2m$ of arrows (with $m\in \Zb/2$ for honeycomb $\On$ and $m \in \Zb$ for TL), which results in a defect $\delta \phi = 2\pi m$ for the height variable.
The effective model for both critical regimes of the {\On} model and for the TL model is thus a compactified boson
\begin{equation} \label{eq:action}
  A[\phi] = \frac{g}{4\pi} \int d^2r\ (\nabla \phi)^2 \,,
  \qquad
  \begin{cases}
    \phi \equiv \phi + \pi & \text{for honeycomb {\On},} \\
    \phi \equiv \phi + 2\pi & \text{for TL.}
  \end{cases}
\end{equation}
The renormalised value of the coupling constant $g$ can be fixed in terms of $n$ using e.g. exact Bethe Ansatz determination of one critical exponent~\cite{Baxter-book}. It is found to be
\begin{equation} \label{eq:g}
  n = -2\cos \pi g \,, \qquad \begin{cases}
    1 < g < 2 & \text{for $K=K_c$,} \\
    0 < g < 1 & \text{for $K>K_c$ or in the TL model.}
  \end{cases}
\end{equation}

For non-contractible loops, the local phase factors defined above cancel each other, so one must introduce an additional topological factor related to the height defects $\delta\phi$ and $\delta'\phi$ along the two cycles of the torus. Simple geometric arguments show that the partition function on the torus is given by~\cite{DSZ87}
\begin{equation} \label{eq:Z-On}
  Z_{{\rm O}(n)} = \sum_{m,m' \in \mathbb{Z}/2}
  \int_{\dblind{\delta\phi=2\pi m}{\delta'\phi=2\pi m'}}
  [D\phi]\ e^{-A[\phi]} \ \cos[\pi e_0 \ (2m) \wedge (2m')] \,,
\end{equation}
where the symbol $\wedge$ stands for the greatest common divider, and
\begin{equation}
  \wt n = 2 \cos \pi e_0 \,.
\end{equation}

For the TL model, we consider a torus consisting of $L \times M$ lattice steps, with $L$ and $M$ both even. Then the defects must be of the form $\delta\phi = 2\pi m$, with $m \in \Zb$. Using the same line of argument as for the honeycomb {\On} model, one gets the partition function
\begin{equation} \label{eq:Z-TL}
  Z_{\rm TL} = \sum_{m,m' \in \mathbb{Z}}
  \int_{\dblind{\delta\phi=2\pi m}{\delta'\phi=2\pi m'}}
  [D\phi]\ e^{-A[\phi]} \ \cos(\pi e_0 \ m \wedge m') \,.
\end{equation}

In both cases, the insertion of the topological factor changes the dependence of $Z$ on the system size, and hence it affects the central charge, which becomes
\begin{equation} \label{eq:cc}
  c = 1 - \frac{6 e_0^2}{g} \,.
\end{equation}
To make contact with the Coulomb-gas notations~\cite{DF84,DF85}, we introduce the background charge $-2\alpha_0 = e_0/\sqrt{g}$, and the central charge is given by $c=1-24\alpha_0^2$.

\begin{figure}
  \begin{center}
    \scalebox{0.5}{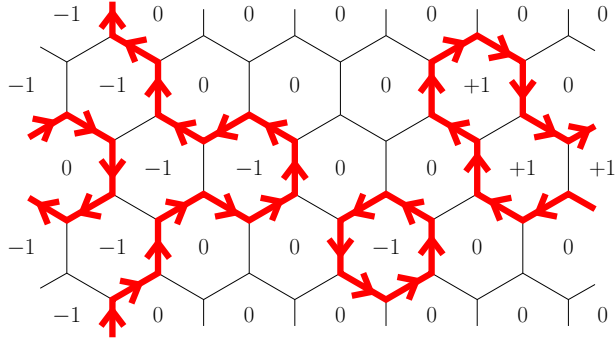}
  \end{center}
  \caption{An oriented loop configuration and associated height values on the dual lattice. Heights are given in units of $\pi$.}
  \label{fig:height}
\end{figure}

\subsection{Vertex operators}
\label{sec:vertex}

To discuss the operator content of the CFT associated to the loop model, we shall use the geometry of the plane, with complex coordinates $(z,\zb)$. The height field $\phi$ can be expanded on holomorphic and anti-holomorphic modes, and one defines the dual field $\theta$ as:
\begin{equation} \label{eq:phi-dual}
  \phi(z,\zb) = \vphi(z) + \vphib(\zb) \,,
  \qquad
  \theta(z,\zb) = \vphi(z) - \vphib(\zb) \,.
\end{equation}
A chiral vertex operator $V_\alpha$, with conformal dimension $h(\alpha)$, is given by
\begin{equation}
  V_\alpha(z) = \ :\exp[i\sqrt{4g}\ \alpha\ \vphi(z)]:
  \qquad
  h(\alpha) = \alpha^2 - 2\alpha \alpha_0 \,,
\end{equation}
where $:\dots:$ denotes normal ordering. A general vertex operator is of the form
\begin{equation} \label{eq:vertex}
  \V_{\alpha,\alphab}(z,\zb) =
  \ :e^{i\sqrt{4g}[\alpha\vphi(z) + \bar\alpha\vphib(\zb)]}:
\end{equation}
Let us recall a heuristic argument relating vertex operators to ``spin-wave'' (electric) and ``vortex'' (magnetic) excitations of the height variable $\phi$. First, if $\alpha=\alphab=\alpha_{\rm el}$, the vertex operator has the form $\V_{\rm el}=\ :\exp(i\sqrt{4g}\ \alpha_{\rm el}\,\phi):$ and is single-valued only if $\sqrt{4g}\alpha_{\rm el}$ is a multiple of the inverse compactification radius :
\begin{equation}
  \alpha_{\rm el} = \frac{k}{\sqrt{4g}} \,,
  \qquad \begin{cases}
    k \in 2\Zb & \text{for honeycomb {\On},} \\
    k \in \Zb & \text{for TL.}
  \end{cases}
\end{equation}
A second interesting case is when $\alpha = 2\alpha_0-\alphab = \alpha_{\rm mag}$. Then the vertex operator has the form 
\begin{equation}
  \V_{\rm mag} =\ :\exp(i\sqrt{4g}\ \alpha_0\ \phi) \times \exp[i\sqrt{4g}\ (\alpha_{\rm mag}-\alpha_0)\theta]:
\end{equation}
where the first factor cancels the background charge, and the second factor creates a defect of height $\delta\phi = 8\pi(\alpha_{\rm mag}-\alpha_0)/\sqrt{4g}$, which must be a (half-)integer multiple of $2\pi$: 
\begin{equation}
  \alpha_{\rm mag} = \alpha_0 - \frac{m\sqrt g}{2} \,,
  \qquad \begin{cases}
    m \in \Zb/2 & \text{for honeycomb {\On},} \\
    m \in \Zb & \text{for TL.}
  \end{cases}
\end{equation}

More generally, we introduce the notation for vertex operator charges
\begin{equation}
  \alpha_{me} = \alpha_0 - \frac{m\sqrt{g}}{2} + \frac{e}{2\sqrt{g}} \,,
  \qquad
  \alphab_{me} = \alpha_0 + \frac{m\sqrt{g}}{2} + \frac{e}{2\sqrt{g}} \,,
\end{equation}
and the corresponding conformal dimensions
\begin{equation} \label{eq:h}
  h_{me} = \frac{1}{4}\left(m\sqrt{g} - \frac{e}{\sqrt g} \right)^2
  - \frac{e_0^2}{4g} \,,
  \qquad
  \hb_{me} = h_{me} + me \,.
\end{equation}
The limit $\nt \to n$ corresponds to $e_0 \to 1-g$, which gives $2\alpha_0 \to \sqrt{g}-1/\sqrt{g}$. Thus if one sets $\alpha_-=\sqrt{g}$ and $\alpha_+=-1/\sqrt{g}$, one gets
\begin{eqnarray}
  2 \alpha_0 \mathop{\longrightarrow}_{\nt \to n}
  \alpha_+ + \alpha_- \,, \qquad \alpha_+\alpha_- = -1 \,, \\
  \alpha_{me} \mathop{\longrightarrow}_{\nt \to n}
  \half(1-m)\alpha_- + \half(1-e)\alpha_+ \,,
\end{eqnarray}
which is, for $m,e$ integers, the Coulomb-gas parameterisation of the Kac table.

\subsection{Operator spectrum}

Let us go back to the toroidal geometry, and consider the transfer matrix (or Hamiltonian) propagating along one cycle of the torus. The transfer matrix conserves the total flux of arrows in the time direction, which we denote $2m$. We distinguish two types of sectors: 
\begin{enumerate}

\item The sector $m=0$ corresponds to the situation where no loop cycles in the transfer direction. It contains the ground state, and the allowed excitations in this sector are of the electric type,
  \begin{equation}
    \alpha=\alphab = \alpha_{0e} \,, \qquad e = e_0 +k\, ,
      \qquad \begin{cases}
    k \in 2\Zb & \text{for honeycomb {\On},} \\
    k \in \Zb & \text{for TL.}
  \end{cases}
 \end{equation}
  and we denote the corresponding operator as $\V_{0e}$. Note that, in the limit $\nt \to n$, the charge $\alpha_{0,e_0+k} \to \alpha_{1,1+k}$, and we obtain the scalar degenerate fields $\V_{0e} \to \V_{1,1+k} = \Phi_{1,1+k}$ with $k \in 2\Zb$ (resp. $k \in \Zb$) for the honeycomb {\On} model (resp. for the TL model).

\item A sector with $m \neq 0$ corresponds to $2m$ strings propagating in the time direction. The lowest energy state in this sector is the ``purely magnetic'' state with $\alpha = 2\alpha_0-\alpha = \alpha_{m0}$, and the allowed excitations are of the form
  \begin{equation}
    (\alpha,\alphab) = (\alpha_{me},\alpha_{-m,e}) \,,
    \qquad \begin{cases}
      m \in \frac{\Zb^+}{2}, e \in \Zb/m & \text{for honeycomb {\On}} \\
      m \in \Zb^+, e \in \Zb/m & \text{for TL} \\
    \end{cases}
  \end{equation}
  and we denote the corresponding operator as $\W_{me}$. This operator has conformal spin $h_{me}-h_{-m,e} = -me$. Note that, for $m>1$, ``fractional'' electric charges, $e=p/m$ with $p$ integer, are allowed.
\end{enumerate}

To summarise, for a transfer matrix with periodic boundary conditions, the states present in the scaling theory are the $\V_{0e}$'s in the zero-string sector, the $\W_{me}$'s in the $2m$-string sector, and their descendants under the action of the Virasoro generators $\{ L_n, \Lb_n \}$.

\section{Conformal bootstrap for mixed electric/magnetic operators}
\label{sec:boot}

Throughout \secref{boot} we will restrict to the case $n=\nt$. For the sake of completeness, and to fix notations, we first review the Coulomb-gas approach~\cite{DF84,DF85} for the case of scalar electric operators, and present a generalisation to mixed electric/magnetic operators, also with the help of functional relations in the fashion of~\cite{Teschner95}.

\subsection{Coulomb-gas approach for purely electric operators}
\label{sec:CG1}

\subsubsection{Four-point function}
We denote $\V_\alpha(z,\zb)=:e^{i\sqrt{4g}\alpha\phi(z,\zb)}:$ a generic electric operator.
Recall that $\V_\alpha$ and $\V_{2\alpha_0-\alpha}$ have the same conformal dimensions $h=\hb=\alpha^2-2\alpha_0\alpha$, and they represent the same operator $\Phi_{h,h}$ in the CFT.
We consider the four-point function
\begin{equation} \label{eq:4pt}
  \C(z,\zb) = \aver{\V_{\alpha_1}(0)\V_{\alpha_2}(z,\zb)\V_{\alpha_{3}}(1)\V_{\alpha_{4}}(\infty)} \,,
\end{equation}
and we assume that the charges $\{ \alpha_i \}$ satisfy the neutrality condition after insertion of $(p-1)$ screening charges $\V_{\alpha_+}$:
\begin{equation} \label{eq:neutral}
  \alpha_1 + \alpha_2 + \alpha_3 + \alpha_4 + (p-1)\alpha_+ = 2\alpha_0 \,.
\end{equation}
Under these conditions, the four-point function can be written
\begin{equation} \label{eq:4ptG}
  \C(z,\zb) = (z\zb)^{2\alpha_1\alpha_2} \ [(z-1)(\zb-1)]^{2\alpha_2\alpha_3} \times G(z,\zb) \,,
\end{equation}
where $G(z,\zb)$ is a bilinear combination of $p$ independent conformal blocks
\begin{equation} \label{eq:G}
  G(z,\zb) = \sum_{k=1}^p \sum_{\kb=1}^p X_{k,\kb}
  \ \F_k(z|\alpha_1, \alpha_2, \alpha_3, \alpha_4)
  \ \ol{\F_{\kb}(z|\alpha_1, \alpha_2, \alpha_3, \alpha_4)} \,.
\end{equation}

\subsubsection{Conformal blocks}
The conformal block $\F_k(z|\alpha_1, \alpha_2, \alpha_3, \alpha_4)$ is defined as the contour integral over the positions $\{v_j\}$ of the $(p-1)$ charges
\begin{equation} \label{eq:Fk}
  \F_k(z|\{\alpha_i\}) =
  \oint_{C_k} \prod_{j=1}^{p-1} dv_j
  \prod_{j=1}^{p-1} v_j^a (v_j-1)^b (v_j-z)^c
  \times \prod_{1\leq i<j\leq p-1} (v_i-v_j)^{2\rho} \,,
\end{equation}
where
\begin{equation}
  a = 2\alpha_+ \alpha_1 \,,
  \quad b = 2\alpha_+ \alpha_3 \,,
  \quad c = 2\alpha_+ \alpha_2 \,,
  \quad d = 2\alpha_+ \alpha_4 \,,
\end{equation}
and $\rho = \alpha_+^2$, and the integration contour $C_k$ is shown in \figref{Ck}a. 
The behaviour of $\F_k$ at $z \to 0$ is
\begin{equation} \label{eq:Fk-lim}
  z^{2\alpha_1\alpha_2} \ \F_k(z) \mathop{\sim}_{z \to 0} N_k \times z^{-h(\alpha_1)-h(\alpha_2)+h(\beta_k)} \,,
\end{equation}
where $\beta_k = \alpha_1+\alpha_2+(k-1)\alpha_+$, and the normalisation factor $N_k$ is given in~\eqref{eq:Nk}. Thus, the conformal block $\F_k$ corresponds to the fusion diagram shown in \figref{fusion}a, and the term $(k,\kb)$ in~\eqref{eq:G} corresponds to the fusions
\begin{equation}
  \V_{\alpha_1} \times \V_{\alpha_2}
  \to \V_{\beta_k,\beta_\kb}
  \leftarrow \V_{\alpha_3} \times \V_{\alpha_4} \,.
\end{equation}

\begin{figure}
  \begin{center}
    \begin{tabular}{ccc}
      \scalebox{1}{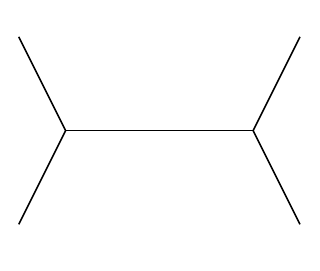}
      &&
      \scalebox{1}{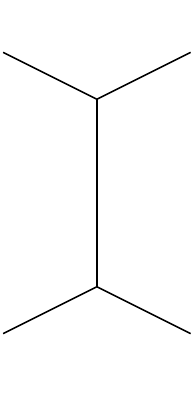} \\
      \\
      (a) &\qquad\qquad\qquad\qquad\qquad& (b)
    \end{tabular}
  \end{center}
  \caption{Fusion diagram corresponding to: (a) the conformal block $\F_k$; (b) the conformal block $\Ft_\ell$.}
  \label{fig:fusion}
\end{figure}

\subsubsection{Dual basis} Another basis of conformal blocks $\Ft_\ell(z|\alpha_1,\alpha_2,\alpha_3,\alpha_4)$ is given by the same integrals as in~\eqref{eq:Fk}, but with the integration contour $C_k$ replaced by $\wt C_\ell$ (see \figref{Ck}b), and they correspond to the fusion diagram of \figref{fusion}b, with the charge $\wt\beta_\ell = \alpha_2+\alpha_3+(\ell-1)\alpha_+$ in the fusion channel. The correlation function may be written
\begin{equation} \label{eq:G2}
  G(z,\zb) = \sum_{\ell,\ellb} \wt X_{\ell,\ellb}
  \ \Ft_\ell(z|\{\alpha_i\}) \ \ol{\Ft_\ellb(z|\{\alpha_i\})} \,.
\end{equation}
The change of basis is given by the matrix $A$ and its inverse $\wt A$:
\begin{equation} \label{eq:change}
  \F_k(z) = \sum_{\ell=1}^p A_{k\ell} \Ft_\ell(z) \,,
  \qquad
  \Ft_\ell(z) = \sum_{k=1}^p \wt A_{\ell k} \F_k(z) \,,
\end{equation}
with the matrix elements of $A$ and $\wt A$ given in~\dblref{eq:Apk}{eq:Akp}.
The expansion coefficients of $G$ in \eqref{eq:G} and \eqref{eq:G2} are related by
\begin{equation}
  \wt X_{\ell,\ellb} = \sum_{k,\kb} A_{k\ell} \ A_{\kb \ellb} \ X_{k,\kb} \,,
  \qquad
  X_{k,\kb} = \sum_{\ell,\ellb} \wt A_{\ell k} \ \wt A_{\ellb \kb} \ \wt X_{\ell,\ellb} \,,
\end{equation}
where we have used the fact that the coefficients $A_{k\ell}$ and $\wt A_{\ell k}$ are real.

\subsubsection{Monodromy conditions}
Under analytic continuation of $\C(z,\zb)$ as $z$ goes around zero,
the term $(k,\kb)$ in the decomposition~\dblref{eq:4ptG}{eq:G} acquires a phase $e^{2i\pi \mu_{k\kb}}$, where [see \eqref{eq:Fk-lim}]
\begin{equation}
  \mu_{k\kb} = h(\beta_k)-h(\beta_\kb)  \,.
\end{equation}
Note that $\mu_{k\kb}$ is simply the conformal spin of the vertex operator $\V_{\beta_k, \beta_\kb}$ in the fusion channel. Similarly, under analytic continuation of $\C(z,\zb)$ as $z$ goes around one, the phase for the term $(\ell,\ellb)$ in~\eqref{eq:G2} is $e^{2i\pi \wt\mu_{\ell\ellb}}$, where
\begin{equation}
  \wt\mu_{\ell\ellb} = h(\wt\beta_\ell)-h(\wt\beta_{\ellb}) \,,
\end{equation}
which is, again, the conformal spin of the operator in the dual fusion channel.

We require $G(z,\zb)$ to be monodromy invariant: this selects the terms in the expansions~\eqref{eq:G} and \eqref{eq:G2}, such that $\mu_{k\kb}$ and $\wt\mu_{\ell\ellb}$ are integers. There are in general several possibilities, depending on the values of the external charges $\{\alpha_i\}$, but we consider here the generic solution, where the selection rules are:
\begin{equation}
  k=\kb \qquad \text{and} \qquad \ell = \ellb \,.
\end{equation}  
In other words, we must have $X_{k,\kb}=0$ for all $k \neq \kb$, and $\wt X_{\ell,\ellb}=0$ for all $\ell \neq \ellb$. This yields the ``diagonal form''
\begin{equation}
  G(z,\zb) = \sum_{k=1}^p X_k \ |\F_k(z|\{\alpha_i\})|^2 \,,
\end{equation}
and the homogeneous linear system
\begin{equation} \label{eq:mono}
  \forall \ell \neq \ellb \,,
  \qquad
  \sum_{k=1}^p A_{k\ell} \ A_{k\ellb} \ X_k = 0 \,.
\end{equation}
The system~\eqref{eq:mono} consists in $p(p-1)$ linearly related equations, for $p$ unknowns $\{X_1, \dots, X_p\}$. It can be shown that~\eqref{eq:mono} is equivalent to the subsystem obtained by taking $\ell \in \{1, \dots, p-1\}$ and $\ellb=p$~:
\begin{equation} \label{eq:mono2}
  \forall \ \ell \in \{1, \dots, p-1\} \,,
  \qquad
  \sum_{k=1}^p A_{k\ell} \ A_{kp} \ X_k = 0 \,,
\end{equation}
which has the simple solution $X_k \propto \wt A_{pk}/A_{kp}$. To get a meaningful set of coefficients, we need to work with the normalised conformal blocks:
\begin{equation} \label{eq:G3}
  G(z,\zb) = \sum_{k=1}^p S_k \ \left|
    \frac{1}{N_k}\F_k(z)
    \right|^2 \,,
    \qquad S_k = |N_k|^2 \times X_k \,.
\end{equation}
After some simple manipulations described in~\appref{app-boot}, we obtain
\begin{align}
  S_k = &
  \prod_{j=0}^{k-2} \gamma[(j+1)\rho] \ \frac{\gamma(1+a+j\rho) \gamma(1+c+j\rho)}{\gamma[2+a+c+(k-2+j)\rho]} \nn \\
  & \times \prod_{j=0}^{p-k-1} \gamma[(j+1)\rho] \ \frac{\gamma(1+b+j\rho) \gamma(1+d+j\rho)}{\gamma[2+b+d+(p-k-1+j)\rho]} \,,
  \label{eq:Sk}
\end{align}
where we have used the notation $\gamma(x) = \Gamma(x)/\Gamma(1-x)$.
Similarly, in the basis $\{ \Ft_\ell \}$, we have the expansion
\begin{equation} \label{eq:G4}
  G(z,\zb) = \sum_{\ell=1}^p \wt S_\ell \ \left|
    \frac{1}{\wt N_\ell}\Ft_\ell(z)
    \right|^2 \,,
\end{equation}
where $\wt S_\ell$ is obtained by changing $k \leftrightarrow \ell$ and $a \leftrightarrow b$ in~\eqref{eq:Sk}.

\subsubsection{Restricted case} 
To lighten the notation, we define $\V_{rs}:=\V_{\alpha_{rs}}$.
Let us consider a particular case of~\dblref{eq:4pt}{eq:neutral}:
\begin{equation} \label{eq:4pt2}
  \C(z,\zb) = \aver{\V_\alpha(0) \V_{2\alpha_0-\alpha}(z,\zb)
    \V_{1p}(1) \V_{1p}(\infty)} \,,
\end{equation}
where $\alpha$ is a generic charge.
Note that, after proper normalisation $\V_\alpha \equiv \V_{2\alpha_0-\alpha}$, but we use the above convention to ensure the neutrality condition~\eqref{eq:neutral}. The term $k$ in~\eqref{eq:G3} corresponds to the fusion channel:
\begin{equation}
  \V_\alpha \times \V_{2\alpha_0-\alpha}
  \to \V_{\beta_k}
  \leftarrow \V_{1p} \times \V_{1p} \,,
\end{equation}
where $\beta_k = 2\alpha_0-\alpha_{1,2k-1}$. The coefficients $S_k$ then read
\begin{align}
  S_k = S_k^{(p)}(a, a') =&
  \prod_{j=0}^{k-2} \gamma[(j+1)\rho] \ \frac{\gamma(1+a+j\rho) \gamma(1+a'+j\rho)}{\gamma[(k+j)\rho]} \nn \\
  & \times \prod_{j=0}^{p-k-1} \gamma[(j+1)\rho] \ \frac{\gamma^2(1-(p-j-1)\rho)}{\gamma[2-(p+k-j-1)\rho]} \,,
  \label{eq:Sk2}
\end{align}
where $a=2\alpha_+ \alpha$, and $a'=2\alpha_+ (2\alpha_0-\alpha)$.

\subsubsection{Extraction of OPE coefficients} 
The expansion coefficients $S_k$ are related to OPE constants by
\begin{equation}
  S_k \propto C(\V_\alpha, \V_\alpha, \V_{1,2k-1}) \times
  C(\V_{1,2k-1}, \V_{1p}, \V_{1p}) \,.
\end{equation}
Using the result~\eqref{eq:Sk2}, and after a few simple steps described in \appref{app-boot}, we obtain the compact form
\begin{equation} \label{eq:C1}
  \boxmath{\begin{aligned}
      C^2(\V_\alpha, \V_\alpha, \V_{1,2k-1}) =&
      \prod_{j=0}^{k-2} \frac{\gamma^3[(j+1)\rho] \gamma^2(1+a+j\rho) \gamma^2(1+a'+j\rho)}{\gamma[(j+k)\rho]} \\
      & \times \prod_{j=0}^{2k-3} \gamma[2-(j+2)\rho] \,.
    \end{aligned}}
\end{equation}
where $a=2\alpha_+ \alpha$, and $a'=2\alpha_+ (2\alpha_0-\alpha)$.
Similarly, the dual coefficients for \eqref{eq:4pt2} give the OPE constants:
\begin{equation}
  \wt S_\ell \propto C^2(\V_{1p}, \V_\alpha, \V_{\alpha+(p+1-2\ell)\alpha_+/2}) \,.
\end{equation}
In this case, the OPE constants are determined up to an overall constant (independent of $q$):
\begin{equation} \label{eq:C2}
  \boxmath{\begin{aligned}
    &C^2(\V_{1p}, \V_\alpha, \V_{\alpha+q\alpha_+})\propto \\
    &\qquad\qquad
    \prod_{j=0}^{\ell-2} \gamma^2[(j+1)\rho]
    \ {\gamma(1+a'+j\rho)}{\ \gamma[1+a+(p-\ell-1-j)\rho]} \\
    &\qquad\qquad
    \times \prod_{j=0}^{p-\ell-1} \gamma^2[(j+1)\rho]
    \ {\gamma(1+a+j\rho)}{\ \gamma[1+a'+(\ell-2-j)\rho]} \,,
  \end{aligned}}
\end{equation}
where $q \in \{-\frac{p-1}{2}, \dots, \frac{p-1}{2}\}$, and $\ell = \frac{p+1}{2}-q$.
An important example is the fusion with $\V_{12}$:
\begin{equation}
  \V_{12} \times \V_\alpha \to \V_{\alpha \pm \alpha_+/2} \,.
\end{equation}
Using~\eqref{eq:C2}, we get
\begin{equation} \label{eq:C2bis}
  \boxmath{
    \frac{C^2(\V_{12}, \V_\alpha, \V_{\alpha-\alpha_+/2})}{C^2(\V_{12}, \V_\alpha, \V_{\alpha+\alpha_+/2})}
    = \frac{\gamma(1+a')\gamma(1+a-\rho)}{\gamma(1+a)\gamma(1+a'-\rho)} \,.
  }
\end{equation}

\subsubsection{Comparison with the DOZZ formula}

As it was noted in~\cite{DV11,PSVD13}, some of the structure constants of CFTs with $c<1$ associated to statistical models satisfy the same functional relations as those of the Liouville theory defined by the Lagrangian
\begin{equation} \label{eq:Liouville}
  \mathcal{L} = \frac{1}{4\pi} (\nabla \phi)^2 + \mu \ e^{2b\phi} \,,
\end{equation}
with the coupling constant $b$ continued to an imaginary value: $b=i\beta$. In this context, the central charge reads:
\begin{equation}
  c = 1-6(\beta-1/\beta)^2 \,.
\end{equation}
The solution of the functional relations (conformal bootstrap) for the Liouville theory is given by~\cite{DO92,DO94,ZZ96,Teschner95}
\begin{align}
  &C_{\cal L}(\V_{\alpha_1},\V_{\alpha_2},\V_{\alpha_3}) = \label{eq:DOZZ} \\
  & \qquad {\rm const} \times
  \frac{\Yb(\beta-\alpha_{12}^3)\Yb(\beta-\alpha_{23}^1)\Yb(\beta-\alpha_{31}^2)
    \Yb(2\beta-\beta^{-1}-\sum_{i=1}^3\alpha_i)}
  {\left[\prod_{i=1}^3 \Yb(\beta-2\alpha_i)\Yb(2\beta-\beta^{-1}-2\alpha_i)\right]^{1/2}} \,,  \nn
\end{align}
where $\alpha_{ij}^k = \alpha_i+\alpha_j-\alpha_k$, and the multiplicative constant is simply determined by $C_{\cal L}(\V_\alpha, \V_\alpha, 1) = 1$. To define the function $\Yb(x)$, it is convenient to introduce $Q=\beta+\beta^{-1}$. On the interval $0<x<Q$, the function $\Yb(x)$ is given by
\begin{equation}
  \Yb(x) = \int_0^\infty \frac{dt}{t} \left[
    (Q/2-x)^2 \ e^{-t} - \frac{\sinh^2 (Q/2-x) \frac{t}{2}}{\sinh \frac{\beta t}{2} \ \sinh \frac{t}{2\beta}}
  \right] \,,
\end{equation}
and the function is fully determined for $x \in \mathbb{R}$ by the relation
\begin{equation}
  \frac{\Yb(x+\beta)}{\Yb(x)} = \gamma(\beta x) \ \beta^{1-2\beta x} \,.
\end{equation}

This solution is usually called the Dorn--Otto--Zamolodchikov--Zamolodchikov (DOZZ) formula.
It can be checked explicitly that the OPE coefficients \dblref{eq:C1}{eq:C2bis}, found from the integral representation of the conformal blocks, coincide with~\eqref{eq:DOZZ}.

\subsection{Coulomb-gas approach in the presence of magnetic operators}
\label{sec:CG2}

We shall generalise slightly the Coulomb-gas approach described above to the case of vertex operators $\V_{\alpha,\alphab}$ with $\alpha \neq \alphab$.

\subsubsection{Four-point function} We consider the correlation function
\begin{equation} \label{eq:4pt3}
  \C(z,\zb) = \aver{\V_{\alpha_1,\alphab_1}(0) \V_{\alpha_2,\alphab_2}(z,\zb) \V_{\alpha_3,\alphab_3}(1) \V_{\alpha_4,\alphab_4}(\infty)} \,,
\end{equation}
subject to the neutrality conditions
\begin{equation}
  \begin{aligned}
    \alpha_1 + \alpha_2 + \alpha_3 + \alpha_4 + (p-1)\alpha_+ &= 2\alpha_0 \,, \\
    \alphab_1 + \alphab_2 + \alphab_3 + \alphab_4 + (p-1)\alpha_+ &= 2\alpha_0 \,.
  \end{aligned}
\end{equation}
We can then write
\begin{equation}  \label{eq:4ptG3}
  \C(z,\zb) = z^{2\alpha_1\alpha_2} \zb^{2\alphab_1\alphab_2} (z-1)^{4\alpha_2\alpha_3} (\zb-1)^{4\alphab_2\alphab_3}
  \times G(z,\zb) \,,
\end{equation}
with
\begin{equation} \label{eq:G5}
  G(z,\zb) = \sum_{k=1}^p \sum_{\kb=1}^p X_{k,\kb}
  \ \F_k(z|\alpha_1,\alpha_2,\alpha_3,\alpha_4)
  \ \ol{\F_\kb(z|\alphab_1,\alphab_2,\alphab_3,\alphab_4)} \,,
\end{equation}
where the $F_k$'s are the same conformal blocks as in~\eqref{eq:G}, but the charges in the holomorphic and anti-holomorphic sectors can be different. In the dual basis we write
\begin{equation} \label{eq:G6}
  G(z,\zb) = \sum_{\ell=1}^p \sum_{\ellb=1}^p \wt X_{\ell,\ellb}
  \ \Ft_\ell(z|\alpha_1,\alpha_2,\alpha_3,\alpha_4)
  \ \ol{\Ft_\ellb(z|\alphab_1,\alphab_2,\alphab_3,\alphab_4)} \,.
\end{equation}
We introduce the notations
\begin{equation}
  \begin{aligned}
    a = 2\alpha_+\alpha_1 \,, \quad b = 2\alpha_+\alpha_3 \,,
    \quad c = 2\alpha_+\alpha_2 \,, \quad d = 2\alpha_+\alpha_4 \,, \\
    \ab = 2\alpha_+\alphab_1 \,, \quad \bar{b} = 2\alpha_+\alphab_3 \,,
    \quad \bar{c} = 2\alpha_+\alphab_2 \,, \quad \bar{d} = 2\alpha_+\alphab_4 \,.
  \end{aligned}
\end{equation}

\subsubsection{Monodromy conditions} 
The term $k,\kb$ in~\eqref{eq:G5} corresponds to the fusions
\begin{equation}
  \V_{\alpha_1,\alphab_1} \times \V_{\alpha_2,\alphab_2} \to  \V_{\beta_k,\betab_k}
  \leftarrow \V_{\alpha_3,\alphab_3} \times \V_{\alpha_4,\alphab_4}
\end{equation}
with the charges in the fusion channel:
\begin{equation}
  \begin{aligned}
    \beta_k &= \alpha_1 + \alpha_2 +(k-1)\alpha_+ \,, \\
    \betab_\kb &= \alphab_1 + \alphab_2 +(\kb-1)\alpha_+ \,.
  \end{aligned}
\end{equation}
Under analytic continuation of~$\C(z,\zb)$ as $z$ goes around zero, this term gets a phase factor $e^{2i\pi \mu_{k,\kb}}$, where
\begin{align}
  \mu_{k,\kb} &= [h(\beta_k)-h(\alpha_1)-h(\alpha_2)] - [h(\betab_\kb)-h(\alphab_1)-h(\alphab_2)] \nn \\
  &=  {\rm spin}(\V_{\beta_k,\betab_k})- {\rm spin}(\V_{\alpha_1,\alphab_1}) - {\rm spin}(\V_{\alpha_2,\alphab_2})  \,.
\end{align}
Similarly, under analytic continuation of~$\C(z,\zb)$ as $z$ goes around one, the term $\ell,\ellb$ in~\eqref{eq:G6} gets a factor $e^{2i\pi \wt\mu_{\ell,\ellb}}$, where
\begin{equation}
  \wt\mu_{\ell,\ellb} =  {\rm spin}(\V_{\wt\beta_\ell,\wt\betab_\ellb})- {\rm spin}(\V_{\alpha_2,\alphab_2}) - {\rm spin}(\V_{\alpha_3,\alphab_3})  \,.
\end{equation}
In general, to construct a monodromy invariant four-point function, one must select the terms such that $\mu_{k,\kb}$ and $\wt\mu_{\ell,\ellb}$ are integers.

\subsubsection{Restricted case}
\label{sec:restrict}
From now on, let us focus on the particular case:
\begin{equation} \label{eq:4pt4}
  \C(z,\zb) = \aver{\V_{\alpha, \alphab}(0) \V_{2\alpha_0-\alpha, 2\alpha_0-\alphab}(z,\zb) \V_{1p}(1) \V_{1p}(\infty)} \,,
\end{equation}
where we require that 
\begin{itemize}
  \item[(i)] the operator $\V_{\alpha, \alphab}$ has integer conformal spin;
  \item[(ii)] the difference $(\alpha-\alphab) \in \Zb \alpha_-/2$.
\end{itemize}
Note that these two conditions are satisfied by the electric operators $\{\V_{1,1+k}\}$ and the $2m$-string excitations $\{\W_{me}\}$. We then have
\begin{equation}
  a+c = \ab + \bar{c} = 2(\rho-1) \,,
  \quad b=\bar{b} =d = \bar{d} = b_p = (1-p)\rho \,,
\end{equation}
and the condition (ii) translates into
\begin{equation}
  a-\ab = \bar{c}-c \in \Zb \,.
\end{equation}
In this context, we can again construct a monodromy invariant by selecting only the terms $k=\kb$ and $\ell=\ellb$. We get a ``diagonal'' decomposition
\begin{equation}
  G(z,\zb) = \sum_{k=1}^p X_k \ \F_k(z|\{\alpha_i\}) \ol{\F_k(z|\{\alphab_i\})} \,,
\end{equation}
with coefficients $X_k$ subject to the linear system
\begin{equation} \label{eq:mono3}
  \forall \ \ell \neq \ellb \,,
  \qquad \sum_{k=1}^p A_{k\ell}(\{\alpha_i\}) A_{k\ellb}(\{\alphab_i\}) \ X_k = 0 \,.
\end{equation}
The above condition (ii) ensures that this system still admits a solution\footnote{
If (ii) is relaxed, then~\eqref{eq:Xk} is a solution of~\eqref{eq:mono3} only for $\ell<\ellb$.
}, which has the form
\begin{equation} \label{eq:Xk}
  X_k \propto \frac{\wt A_{pk}(\{\alpha_i\})}{A_{kp}(\{\alphab_i\})} \,.
\end{equation}
Introducing the coefficients $S_k$ on normalised conformal blocks:
\begin{equation}
  G(z,\zb) = \sum_{k=1}^p S_k \ \frac{\F_k(z|\{\alpha_i\})}{N_k(\{\alpha_i\})}
  \times
  \ol{\frac{\F_k(z|\{\alphab_i\})}{N_k(\{\alphab_i\})}} \,,
\end{equation}
we obtain from~\eqref{eq:Sk}:
\begin{align}
  S_k = \sqrt{S_k^{(p)}(a,a') \ S_k^{(p)}(\ab, \ab')} \,,
  \label{eq:Sk3}
\end{align}
where
\begin{equation}
  a = 2\alpha_+ \alpha \,,
  \quad a'=2\alpha_+(2\alpha_0-\alpha) \,,
  \quad \ab = 2\alpha_+ \alphab \,,
  \quad \ab' = 2\alpha_+ (2\alpha_0-\alphab) \,,
\end{equation}
and $S_k^{(p)}(a,a')$ is defined in~\eqref{eq:Sk2}.

\subsubsection{Extraction of OPE coefficients}
The above expansion coefficients give access to the OPE constants through the relation
\begin{equation}
  S_k \propto C(\V_{\alpha,\alphab},\V_{\alpha,\alphab},\V_{1,2k-1}) \times C(\V_{1,2k-1},\V_{1p},\V_{1p}) \,.
\end{equation}
Using the results for purely electric operators, we get 
\begin{equation} \label{eq:C3}
  \boxmath{
    C(\V_{\alpha,\alphab},\V_{\alpha,\alphab},\V_{1,2k-1}) =
    \sqrt{C(\V_\alpha,\V_\alpha,\V_{1,2k-1}) \ C(\V_\alphab,\V_\alphab,\V_{1,2k-1})} \,,
  }
\end{equation}
where $C(\V_\alpha,\V_\alpha,\V_{1,2k-1})$ is the purely electric coefficient given in~\eqref{eq:C1}.
Similarly, using the dual expansion we obtain
\begin{equation} \label{eq:C4}
  \boxmath{
    C(\V_{1p}, \V_{\alpha,\alphab}, \V_{\alpha+q\alpha_+,\alphab+q\alpha_+}) \propto
    \sqrt{C(\V_{1p}, \V_\alpha, \V_{\alpha+q\alpha_+}) \ C(\V_{1p}, \V_\alphab, \V_{\alphab+q\alpha_+})} \,,
  }
\end{equation}
where $q \in \{-\frac{p-1}{2}, \dots, \frac{p-1}{2}\}$, and $C(\V_{1p}, \V_\alpha, \V_{\alpha+q\alpha_+})$ was given in~\eqref{eq:C2}. In particular, for $p=2$, we get
\begin{equation} \label{eq:C5}
  \boxmath{
    \frac{C^2(\V_{12}, \V_{\alpha,\alphab}, \V_{\alpha-\frac{\alpha_+}{2},\alphab-\frac{\alpha_+}{2}})}
    {C^2(\V_{12}, \V_{\alpha,\alphab}, \V_{\alpha+\frac{\alpha_+}{2},\alphab+\frac{\alpha_+}{2}})}
    = \sqrt{\frac{\gamma(1+a')\gamma(1+\ab')\gamma(1+a-\rho)\gamma(1+\ab-\rho)}
      {\gamma(1+a)\gamma(1+\ab)\gamma(1+a'-\rho)\gamma(1+\ab'-\rho)}} \,.
  } 
\end{equation}
Like for purely electric operators, these coefficients are simply related to the DOZZ formula~\eqref{eq:DOZZ}:
$C(\V_{\alpha_1,\alphab_1},\V_{\alpha_2,\alphab_2},\V_{\alpha_3,\alphab_3}) = \sqrt{C_{\cal L}(\V_{\alpha_1},\V_{\alpha_2},\V_{\alpha_3})C_{\cal L}(\V_{\alphab_1},\V_{\alphab_2},\V_{\alphab_3})}$.
However, note that this simple form of the OPE coefficients~\dblref{eq:C3}{eq:C5} as geometric means of the holomorphic and anti-holomorphic parts is only valid under the assumptions (i) and (ii) of \secref{restrict}. For instance, the OPE coefficient of purely magnetic operators $C(\W_{10},\W_{10},\W_{10})$ is expected to be finite, whereas the formula~\eqref{eq:DOZZ} with charges $\alpha_i=\alpha_{10}$ gives an infinite coefficient.

\subsection{Functional relation on OPE coefficients}
\label{sec:funct}

\subsubsection{General setting}
Throughout this section we move away from the  {\On} loop model and consider a generic CFT with a central charge of the form 
\begin{align}
  c = 1 - 6 \frac{(1-g)^2}{g} = 1 -24 \alpha_0^2
\end{align}
with $g$ irrational. We the following assumptions about this CFT:
\begin{itemize}
\item [(i)] There are no degeneracies in the spectrum.
\item [(ii)] The spectrum contains a scalar field $\Phi_{12}(z,\zb)$ with conformal dimension $h_{12} = \bar{h}_{12} = (3-2g)/4g$. This field is degenerate at level $2$, namely the descendents
  \begin{align}
    \chi_{12} = \left(L_{-2} - g L_{-1}^2 \right)  \Phi_{12}
    \qquad\text{and}\qquad
    \bar{\chi}_{1,2} = \left(\bar{L}_{-2} - g \bar{L}_{-1}^2 \right) \Phi_{12}
  \end{align}
  have zero norm.
\item [(iii)] Finally the last assumption we make is that both $\chi_{12}$ and $\bar\chi_{12}$ decouple from the theory. Note that, in unitary CFTs, this last assumption follows from the previous ones, but not in logarithmic CFTs, where a null vector can still have non-zero overlap with its logarithmic partner.
\end{itemize}
It is well known that the decoupling of $\chi_{12}$ yields a differential equation for any correlation function that involves the field $\Phi_{12}$. In particular it yields a functional equation for the structure constants of scalar fields~\cite{Teschner95}. In this section we extend this result to the case of primary fields $\Phi_{h,\hb}$ with arbitrary spins. The method is essentially the same as in~\cite{Teschner95}, and relies on locality and crossing symmetry of the four-point functions.

\subsubsection{Four-point function involving $\Phi_{12}$}
Let us consider a four-point function of primary fields 
\begin{align} \label{4-point function}
  G(z,\bar{z}) = \aver{\Phi_{h_1,\hb_1}(0) \Phi_{12}(z,\bar{z}) \Phi_{h_3,\hb_3}(1)  \Phi_{h_4,\hb_4}(\infty)} \,.
\end{align}
The decoupling of $\chi_{12}$ implies the following (holomorphic) differential equation:
\begin{align} \label{PDE}
  \left[ g\partial_z^2 + \left( \frac{1}{z} + \frac{1}{z-1}\right)\partial_z + \left(\frac{h_1}{z} - \frac{h_3}{z-1} + h_2- h_4 \right) \frac{1}{z(z-1)} \right] G(z,\bar{z}) = 0 \,.
\end{align}
Likewise, the decoupling of $\bar{\chi}_{12}$ yields a similar antiholomorphic PDE, with $h_i \to \hb_i$, $z \to \zb$, $\partial_z \to \partial_{\zb}$. In order to write down the solutions of these differential equations, it is convenient to parametrise the conformal dimensions {\it \`a la} Coulomb gas, namely  
\begin{align}
  h = \alpha (\alpha-2\alpha_0), \qquad \bar{h} =  \alphab (\alphab- 2\alpha_0)
\end{align} 
We stress that this is purely a convenient parametrisation of the conformal dimensions $(h,\bar{h})$, and that we are in no way working with a Coulomb gas CFT. Note that the two values $\alpha$ and $\alpha' = 2\alpha_0 - \alpha$ lead to the same conformal dimension $h(\alpha) =  -\alpha  \alpha'$. As a consequence, there are four ways to parametrise a field with given left and right conformal dimensions $h=h(\alpha)$, $\hb = \hb(\alphab)$ : 
\begin{align}
  (\alpha,\alphab), \quad (\alpha',\alphab), \quad (\alpha,\alphab'), \quad (\alpha',\alphab') \,.
\end{align}
In the following we will write $\Phi_{\alpha,\alphab}$ as a short-hand for $\Phi_{h(\alpha),\hb(\alphab)}$. 
At this point we put no constraint on the spectrum, but it will turn out that the decoupling of $\chi_{12}$ and $\bar{\chi}_{12}$ severely constrains the admissible values of $(\alpha,\alphab)$. 

\subsubsection{Conformal blocks}
Let us focus for now on the holomorphic conformal blocks, which are given by the two independent solutions of the PDE~\eqref{PDE}. The Coulomb gas parametrisation of the conformal dimensions is very convenient to write down explicit expressions for the conformal blocks. For this purpose we adopt the same notations as in~\secref{CG2}:
\begin{equation}
  \begin{aligned}
    &&a = 2\alpha_+ \alpha_1  \,, \quad && b = 2\alpha_+ \alpha_3  \,, \quad && d = 2\alpha_+ \alpha_4 \,, \\
    &&a'= 2\alpha_+ \alpha'_1 \,, \quad && b'= 2\alpha_+ \alpha'_3 \,, \quad && d'= 2\alpha_+ \alpha'_4 \,.
  \end{aligned}
\end{equation}
We will restrict ourselves to the generic situation in which the exponents at the singularities of the PDE  do not differ by an integer. This means $a-a'$, $b-b'$ and $d-d'  \notin 2\Zb$. This ensures that the conformal blocks have no logarithm in their expansion\footnotemark \footnotetext{In the Appendix, we comment on the case $a-a' \in 2\Zb $, for which logarithms appear in the conformal blocks. In the loop model at $\nt =n$ this occurs for operators $\W_{m,0}$, with $m \in \Zb$.}.  

Several basis of conformal blocks can be considered. In the following we will use two such bases. The first one $\{ I_1(z), I_2(z) \}$ corresponds to the expansion around the singularity $z=0$ : 
\begin{equation}
  \begin{aligned}
    I_1(z) & =  z^{-\frac{a}{2}}(1-z)^{-\frac{b}{2}}\ {}_2F_{1}\left(
      {\textstyle \frac{d-a-b+\rho}{2}, \frac{d'-a-b+\rho}{2}; 1 + \frac{a' - a}{2}; z}
    \right) \,, \\ 
    I_2(z) & = z^{-\frac{a'}{2}}(1-z)^{-\frac{b'}{2}}\ {}_2F_{1}\left(
      {\textstyle \frac{d'-a'-b'+\rho}{2}, \frac{d-a'-b'+\rho}{2}; 1 + \frac{a - a'}{2}; z}
    \right) \,,
  \end{aligned}
\end{equation}
and they correspond to the following holomorphic fusion rules
\begin{equation}
  \begin{aligned}
    I_1(z) : \qquad&  \Phi_{\alpha_1} \times \Phi_{12} \to  \Phi_{\alpha_1 - \frac{\alpha_+}{2}} \,, \\ 
    I_2(z) : \qquad&  \Phi_{\alpha_1} \times \Phi_{12} \to  \Phi_{\alpha_1 + \frac{\alpha_+}{2}} \,.
  \end{aligned}
\end{equation}
The constraint $a-a' \notin 2\Zb$ ensures that the two fields produced in the fusion $\Phi_{\alpha_1} \times \Phi_{12}$ do not have conformal dimensions that differ by an integer. In turn this imposes that the two conformal blocks $I_1(z)$ and $I_2(z)$ have different monodromies around $z=0$:
\begin{equation} \label{eq:mono-Ik}
  I_1(z) \to e^{-i\pi a} I_1(z) \,,
  \qquad I_2(z) \to e^{-i\pi a'} I_2(z) \,.
\end{equation}
While $I_1(z)$ and $I_2(z)$ have Abelian monodromy around $z=0$, they have a more complicated monodromy around $z=1$. Another convenient basis $\{J_1(z), J_2(z)\}$ corresponds to the expansion around the singularity $z=1$, and has Abelian monodromy around it:
\begin{equation}
  J_1(z) \to e^{-i\pi b} J_1(z) \,,
  \qquad J_2(z) \to e^{-i\pi b'} J_2(z) \,.
\end{equation}
Their explicit expressions, which can be found in~\appref{app-hyper}, are readily obtained from $I_1(z)$ and $I_2(z)$ by exchanging $z \leftrightarrow (1-z)$ and $(a,a') \leftrightarrow (b,b')$. These two bases are related through
\begin{equation}
  I_k(z) = \sum_{\ell=1}^2 M_{k\ell}(a,b,d) \ J_\ell(z) \,,
\end{equation}
where the explicit expression of the matrix $M(a,b,c)$ is given in~\appref{app-hyper}. For the antiholomorphic side, we have similar expressions for the bases $\{ \bar I_1(\zb), \bar I_2(\zb) \}$ and $\{ \bar J_1(\zb), \bar J_2(\zb) \}$, with $\alpha_i \to \bar{\alpha}_i$. Since it obeys both PDEs, the full four-point function $G(z,\zb)$ has to be of the form  
\begin{equation}
  G(z,\zb) =  \sum_{k,\bar k} X_{k \bar k}\ I_k(z)\ \bar I_{\bar k}(\zb) \,.
\end{equation}

\subsubsection{Consequences of locality}
Locality enforces the function $G(z,\zb)$ to have a monodromy around $z=0$ of the form
\begin{align}
  G(z,\zb) \to  e^{2i\pi \mu}\ G(z,\zb) \,.
\end{align} 
Note that this phase factor is non-trivial when the fields $\Phi_{\alpha_1,\alphab_1}$ and $\Phi_{12}$ are mutually semi-local. A typical exemple would be the spin and disorder operators for the Ising model, for which $e^{2i\pi\mu} = -1$.
Since we consider the case when $a-a' \notin 2\Zb$ and $\bar{a}-\bar{a}' \notin 2\Zb$, the equality of phase factors $e^{i\pi(\bar a-a)}=e^{i\pi(\bar a-a')}=e^{i\pi(\bar a'-a)}=e^{i\pi(\bar a'-a')}$ can only be satisfied if $X_{k\bar k}$ is a diagonal or anti-diagonal matrix. Since the change $\alpha_1 \to \alpha_1'$ (resp. $\bar{\alpha}_1 \to \bar{\alpha}_1'$) is simply a reparameterisation of the field $\Phi_{\alpha_1,\bar{\alpha}_1}$ which exchanges the roles of  $I_1(z)$ and $I_2(z)$ (resp. $\bar I_1(z)$ and $\bar I_2(z)$), we can assume without any loss of generality a diagonal decomposition 
\begin{align}
  G(z,\zb) =  X_1 \,  I_1(z)\bar I_1(\zb)  + X_2 \, I_2(z)\bar I_2(\zb) \,, \label{X_decomp}
\end{align}
where we have renamed $X_{11}$ to $X_1$ and $X_{22}$ to $X_2$. These two constants are closely related to the OPE structure constants:
\begin{equation} \label{eq:X-OPE}
  \begin{aligned}
    X_1 & = A\ C\left(\Phi_{12},\Phi_{\alpha_1,\alphab_1}, \Phi_{\alpha_1 - \frac{\alpha_+}{2},\alphab_1- \frac{\alpha_+}{2}}\right)
    C\left(\Phi_{\alpha_3,\alphab_3}, \Phi_{\alpha_4,\alphab_4}, \Phi_{\alpha_1 - \frac{\alpha_+}{2},\alphab_1- \frac{\alpha_+}{2}}\right) \,,\\
    X_2 & = A\  C\left(\Phi_{12},\Phi_{\alpha_1,\alphab_1}, \Phi_{\alpha_1 + \frac{\alpha_+}{2},\alphab_1+ \frac{\alpha_+}{2}}\right)
    C\left(\Phi_{\alpha_3,\alphab_3}, \Phi_{\alpha_4,\alphab_4}, \Phi_{\alpha_1 + \frac{\alpha_+}{2},\alphab_1+ \frac{\alpha_+}{2}}\right) \,,
  \end{aligned}
\end{equation}
where $A$ is an unknown constant.

It follows that the monodromy between $\Phi_{12}$ and an arbitrary field $\Phi_{h,\hb}$ can only be $e^{2i\pi \mu} = \pm 1$. Indeed, the monodromy around $z=0$ is  
\begin{align}
  G(z,\zb)  \to   e^{i\pi (\ab-a)} X_1 \,  I_1(z)\bar I_1 (\zb)  +  e^{i\pi (\ab'-a')} X_2 \, I_2(z)\bar I_2 (\zb) \,, 
\end{align}
and generically we expect both $X_1$ and $X_2$ to be non zero. This requires $2\mu \equiv a-\ab \equiv a' - \ab'$ mod $2$, which boils down to\footnotemark \footnotetext{This comes from our choice of working with a diagonal matrix $X$. If one works with an anti-diagonal $X$ instead the conditions becomes $2\mu = a-\ab' = a' - \ab$ mod $2$, or equivalently $a' -\bar{a} \in \Zb$. But as we already mentioned this is a pure matter of convention.} 
\begin{equation} \label{eq:cond-charge}
  \boxed{a - \bar{a} \in \Zb \,,} 
\end{equation}
and therefore $\mu \in \Zb/2$. We prove in the Appendix that this condition has to be obeyed by all primary fields. Let us stress out the meaning of this relation. The decoupling of the null vectors of $\Phi_{12}$ puts strong constraints on the spectrum of the CFT : the conformal dimensions of the primary fields $\Phi_{h,\hb}$ must be of the form
\begin{align}
  \boxed{
    h= \alpha(\alpha -2 \alpha_0), \qquad \bar{h} = \beta (\beta - 2\alpha_0)\,,
    \qquad \text{with} \qquad
    (\alpha - \beta) \in \frac{\Zb \alpha_-}{2} \,.
  }
\end{align}
While this holds trivially for spinless fields, for which we can choose $\alpha = \beta$, this is a severe constraint for non scalar fields. The full spectra of both the TL and {\On} loop models satisfy this condition, although $\Phi_{12}$ does not belong to the spectrum of {\On}. This is perhaps not so surprising. Indeed, the {\On} loop model contains the field $\Phi_{13}$, and the corresponding differential equation naively yields the constraint : $h(\alpha + \alpha_+) - h(\alphab + \alpha_+) \equiv h(\alpha) - h(\alphab) \equiv h(\alpha - \alpha_+) - h(\alphab - \alpha_+) $ mod $1$, which boils down to $a-\ab \in \Zb$. More generally, we expect the differential equation corresponding to $\Phi_{1k}$ to be consistent with locality as long as $a-\ab$ is integer.  From now on, we assume that all fields obey this constraint, namely $(\alpha - \alphab)\in {\Zb\alpha_-}/{2}$ [or $(\alpha' - \alphab)\in {\Zb\alpha_-}/{2}$]. 

As a side remark, if a CFT contains both $\Phi_{12}$ and $\Phi_{21}$ (and if their null vectors decouple), fields in the spectrum must belong to one of the two following families:
\begin{itemize}
\item Scalar fields $\alpha = \alphab$, with $\alpha$ arbitrary. This is the case for all fields in Liouville.
\item Spinful fields $\alpha = \alpha_{m,e}$, $\alphab = \alpha_{-m,e}$, with $e,m \in \Zb/2$, with spin $s = -m\,e \in \Zb/4$.
\end{itemize} 
While a continuum of scalar fields is possible, only a discrete set of spinful fields is allowed in that case. Let us stress that the spin fields of the loop models under consideration do not fall into this classification, since these models contain fields with $m= - \bar{m} \neq 0$ and $e \notin \Zb/2$. This is allowed because $\Phi_{21}$ is absent from the spectrum of the loop models. 

\subsubsection{Solution of the monodromy problem}
Let us go back to the four-point function $G(z,\zb)$ of \eqref{X_decomp}. Repeating the same arguments around $z=1$  yields a decomposition on the $\{J_1(z), J_2(z) \}$ basis : 
\begin{align}
  G(z,\zb) = Y_1 \, J_1(z)\bar J_1(\zb)  + Y_2 \, J_1(z)\bar J_2(\zb) \,.
\end{align}
Compatibility with the change of bases yields two consistency relations :
\begin{equation} \label{X12}
  \begin{aligned}
    M_{11} \bar{M}_{12} \, X_1 +  M_{21} \bar{M}_{22} \, X_2 = 0 \,, \\
    M_{12} \bar{M}_{11} \, X_1 +  M_{22} \bar{M}_{21} \, X_2 = 0 \,.
  \end{aligned}
\end{equation}
For the the two linear equations \eqref{X12} to have a non-zero solution $(X_1,X_2)$ (and therefore a non-vanishing four-point function), we must have
\begin{align}
  \det \left( \begin{array}{cc}  M_{11} \bar{M}_{12} &  M_{21} \bar{M}_{22} \\ 
      M_{12} \bar{M}_{11} & M_{22} \bar{M}_{21} \end{array} \right) =0 
  \qquad\Leftrightarrow\qquad
  \frac{M_{11}M_{22}}{M_{12}M_{21}} = \frac{\bar{M}_{11}\bar{M}_{22}}{\bar{M}_{12}\bar{M}_{21}} \,,
\end{align}
or more explicitly 
\begin{align}
  \frac{ \sin \pi \left( \frac{\rho+a-b-d}{2} \right)\sin \pi \left( \frac{\rho+b-a-d}{2} \right) }
  {\sin \pi \left( \frac{\rho+d-a-b}{2} \right)\sin \pi \left( \frac{3\rho-d-a-b}{2} \right)} 
  = \frac{ \sin \pi \big( \frac{\rho+\ab-\bb-\db}{2} \big)\sin \pi \big( \frac{\rho+\bb-\ab-\db}{2} \big) }
  {\sin \pi \big( \frac{\rho+\db-\ab-\bb}{2} \big)\sin \pi \big( \frac{3\rho-\db-\ab-\bb}{2} \big)} \,,
\end{align}
where $a-\ab$, $b-\bb$ and $d-\db$ are all integers. This constraint is manifestly symmetric under exchange of $a$ and $b$, but not under exchange of $a$ and $d$. This is because our analysis has singled out the relation between the conformal blocks around $z=0$ and $z=1$. Repeating the same steps around $z=\infty$, we would end up with two more relations: one symmetric under $a \leftrightarrow d$ and another symmetric under $b \leftrightarrow d$. When $a,b$ and $d$ are generic, all these constraints boil down to 
\begin{equation}
  (\ab+\bb +\db) \equiv (a+b+d) \mod 2 \,.
\end{equation}
The solution can then be written as
\begin{align}
  \frac{X_1}{X_2} = - \frac{M_{21} \bar{M}_{22}}{M_{11} \bar{M}_{12}} 
  = - \frac{M_{22} \bar{M}_{21}}{M_{12} \bar{M}_{11} } \,, \label{X1overX2}
\end{align}
or in a more symmetric fashion
\begin{align}
\left(\frac{X_1}{X_2}\right)^2 =  \frac{M_{22}M_{21} }{M_{11} M_{12}} \, \frac{\bar{M}_{22} \bar{M}_{21}}{\bar{M}_{11}\bar{M}_{12}  } \,.
\end{align}
This leads to
\begin{align} \label{eq:X-ratio}
\left( \frac{X_1}{X_2} \right)^2 & =    \frac{\gamma \left(\frac{a-a'}{2}\right)}{\gamma \left(\frac{a'-a}{2} \right)} \frac{\gamma \left( \frac{d-a-b'+\rho}{2} \right)\gamma \left( \frac{d-a-b+\rho}{2} \right)}{\gamma \left( \frac{d-a'-b'+\rho}{2} \right)\gamma \left( \frac{d-a'-b+\rho}{2} \right)} \times \left( \alpha_i \to \alphab_i \right) \,.
\end{align}
This expression is invariant under $b \leftrightarrow b'$, $d \leftrightarrow d'$ or $b \leftrightarrow d$, as it should be. In particular, taking $b = \bb = -\rho$ and $(d,\db) = (a,\ab)$ we recover equation \eqref{eq:C5}, obtained in the framework of the Coulomb gas. 

\subsubsection{OPE coefficients}
Using~\eqref{eq:C5}, \eqref{eq:X-OPE} and \eqref{eq:X-ratio}, we can extract a ratio of OPE coefficients, and if we denote $\Phi_2=\Phi_{\alpha_3,\alphab_3}$ and $\Phi_3=\Phi_{\alpha_4,\alphab_4}$, we get
\begin{equation} \label{eq:C6}
  \boxmath{\begin{aligned}
      &\left[\frac{C\left(\Phi_{\alpha_1 - \frac{\alpha_+}{2},\alphab_1- \frac{\alpha_+}{2}}, \Phi_2, \Phi_3 \right)}
        {C\left(\Phi_{\alpha_1 + \frac{\alpha_+}{2},\alphab_1+  \frac{\alpha_+}{2}}, \Phi_2, \Phi_3\right) }\right]^2 \\
      & \qquad = \sqrt{\frac{\gamma \left(\frac{a-a'}{2}\right)\gamma(1+a)}{\gamma \left(\frac{a'-a}{2} \right)\gamma(1+a')}}
      \ \frac{\gamma \left( \frac{d-a-b'+\rho}{2} \right)\gamma \left( \frac{d-a-b+\rho}{2} \right)}{\gamma \left( \frac{d-a'-b'+\rho}{2} \right)\gamma \left( \frac{d-a'-b+\rho}{2} \right)} \times \left( \alpha_i \to \alphab_i \right) \,.
    \end{aligned}}
\end{equation}
This formula immediately applies to the loop model CFT, since the latter satisfies the assumptions made in the beginning of this section. Like in~\eqref{eq:C5}, although it is valid in the presence of electromagnetic operators $\W_{me}$, it only gives a recursion relation on the electric charge.

Equation~\eqref{eq:C6} is a functional relation satisfied by the OPE coefficients. One can check that the RHS of~\eqref{eq:C6} is what one would get if one replaced $C( \dots\Phi_{\alpha_i,\bar\alpha_i}\dots)$ by $\sqrt{C_{\cal L}( \dots\V_{\alpha_i}\dots)C_{\cal L}( \dots\V_{\alphab_i}\dots)}$. However, unlike the Liouville theory where two functional relations (coming from the degeneracy of $\Phi_{12}$ and $\Phi_{21}$) determine uniquely the solution $C_{\cal L}$, in our situation the field $\Phi_{21}$ is absent from the spectrum, and we only have one functional relation, which only determines the OPE coefficients up to a periodic factor.

\subsubsection{Application to the loop model} When specialised to the TL loop model, \eqref{eq:C6} yields, in particular, the ratio
\begin{align}
  \frac{C(\W_{10},\W_{10},\W_{12})}{C(\W_{10},\W_{10},\W_{10})}
  &= \left[
    \frac{\gamma(1-\rho)\gamma(-1-\rho)}
    {\gamma(-1+2\rho)\gamma(1+2\rho)\gamma(-1+\rho)\gamma(1+\rho)}
  \right]^{1/4} \nn \\
  & \qquad \times 
  \left[ \frac{\gamma^3(\rho/2)\gamma(-1+\rho/2)}{\gamma^3(-\rho/2)\gamma(-1-\rho/2)} \right]^{1/2} \,,
  \label{eq:C-theo4}
\end{align}
which is checked numerically in~\secref{num}. Besides the precise expression of the ratio, it is important to notice the non-vanishing of the numerator:
\begin{equation}
  C(\W_{10}, \W_{10}, \W_{12}) \neq 0
  \qquad \text{for} \quad 0<n<2 \,.
\end{equation}
Indeed, since $\W_{12}$ has conformal dimension $h_{12}$ in the holomorphic sector, one could have expected that it obeys fusion rules of the form $\W_{12} \times \V_{\alpha,\alphab} \to \V_{\beta,\betab}$, with $\beta = \alpha \pm \alpha_+/2$, whereas this rule is clearly violated by $C(\W_{10}, \W_{10}, \W_{12})$. As we shall see in~\secref{LCFT}, this is due to the presence of a logarithmic partner for $\W_{10}$.

\subsection{Numerical checks}
\label{sec:num}

Using a numerical algorithm based on transfer-matrix diagonalisation (see~\appref{app-num}), we have computed numerically the OPE coefficients in the TL loop model at the point $n=\nt$, for $0<n<2$. We can then compare these data to our predictions~\eqref{eq:C3}, \eqref{eq:C4}, \eqref{eq:C5} and , \eqref{eq:C6}, in particular in the presence of watermelon operators $\W_{me}$.

\subsubsection{The OPE coefficient~$C(\W_{10}, \W_{10}, \V_{13})$} This OPE coefficient involves two spinless operators together with a degenerate operator of the form $\V_{1,2k-1}$, and thus one expects it to be given by ~\eqref{eq:C1}:
\begin{equation} \label{eq:C-theo1}
  C(\W_{10}, \W_{10}, \V_{13}) = \sqrt{\frac{\gamma^3(\rho)\gamma(2-2\rho)\gamma(2-3\rho)}{\gamma(2\rho)}}
  \ \gamma(\rho+1) \gamma(\rho-1) \,.
\end{equation}
The agreement with numerical data on cylinders of circumference $L=10$ sites to $L=18$ sites is good, as shown in~\figref{C_w10_w10_v13}.

\subsubsection{The OPE coefficient~$C(\W_{11}, \W_{11}, \V_{13})$} In this case we use the prediction of~\eqref{eq:C3} with $k=2$:
\begin{equation} \label{eq:C-reg}
  C(\V_{\alpha,\alphab}, \V_{\alpha,\alphab}, \V_{13}) =
  \sqrt{\frac{\gamma^3(\rho)\gamma(2-2\rho)\gamma(2-3\rho)}{\gamma(2\rho)}
  \gamma(1+a)\gamma(1+a')\gamma(1+\ab)\gamma(1+\ab')} \,,
\end{equation}
where
\begin{equation}
  a = 2\alpha_+ \alpha \,, \quad
  \ab = 2\alpha_+ \alphab \,, \quad
  a' = 2\alpha_+ (2\alpha_0-\alpha) \,, \quad
  \ab' = 2\alpha_+ (2\alpha_0-\alphab) \,.
\end{equation}
The expression~\eqref{eq:C-reg} is ill-defined when $\alpha \to \alpha_{11}$ and $\alphab \to \alpha_{-1,1}$, and hence we need a regularisation procedure to extract $C(\W_{11}, \W_{11}, \V_{13})$.
Indeed, if we set $\alpha = \alpha_{11} + \epsilon$ and $\alphab = \alpha_{-1,1} + \bar\epsilon$ with $\epsilon,\bar\epsilon \to 0$, we get
$
\gamma(1+a)\gamma(1+\ab) \sim \epsilon/\bar\epsilon
$.
The condition~\eqref{eq:cond-charge} imposed by locality of the four-point functions yields a unique possible relation $\epsilon=\bar\epsilon$, which finally leads to
\begin{equation} \label{eq:C-theo2}
  C(\W_{11}, \W_{11}, \V_{13}) = \sqrt{\frac{\gamma^3(\rho) \rho(2-3\rho)\gamma(1+2\rho)}{\gamma(2\rho)}} \,.
\end{equation}
The agreement with numerical data is good, as shown in~\figref{C_w11_w11_v13}. The numerical data show an apparent divergence at $n=1$, but the height of the corresponding peak decays as $L$ grows. We interpret this peak as a numerical artefact arising from a level crossing in the zero-leg sector with zero momentum, between the states $\ket{\V_{13}}$ and $L_{-2} \bar L_{-2}\ket 0$.

\subsubsection{The ratio $C(\V_{12},\W_{11},\W_{12})/C(\V_{12},\W_{11},\W_{10})$} This ratio is predicted by~\eqref{eq:C5}, using the same regularisation procedure for $\W_{11}$ as above:
\begin{equation} \label{eq:C-theo3}
  \frac{C(\V_{12},\W_{11},\W_{12})}{C(\V_{12},\W_{11},\W_{10})}
  = \left[
    \frac{\gamma(-1+2\rho)\gamma(1+2\rho)\gamma(1-\rho)\gamma(-1-\rho)}
    {\gamma(-1+\rho)\gamma(1+\rho)}
  \right]^{1/4} \,.
\end{equation}
The agreement with numerical data is good, as shown in~\figref{ratio_C_v12_w11_w12_C_v12_w11_w10}. In this case, although $h_{12}+h_{-1,2}< 2h_{10}+4$ on the whole interval $0<n<2$, in finite size some level crossing occursat $n=n^*$ between the lattice representatives of $\ket{\W_{12}}$ and some fourth-level descendents of $\ket{\W_{10}}$. This affects the numerical data close to $n=0$, but the crossing location $n^* \to 0$ as the system size grows.

\subsubsection{The ratio $C(\W_{10},\W_{10},\W_{12})/C(\W_{10},\W_{10},\W_{10})$} This ratio is predicted by~\eqref{eq:C-theo4}. The agreement with numerical data is good, as shown in~\figref{ratio_C_w10_w10_w12_C_w10_w10_w10}, although the data still suffer from the finite-size level crossing of $\ket{\W_{12}}$ with the descendents of $\ket{\W_{10}}$ discussed above.

\begin{figure}
  \begin{center}
    \includegraphics{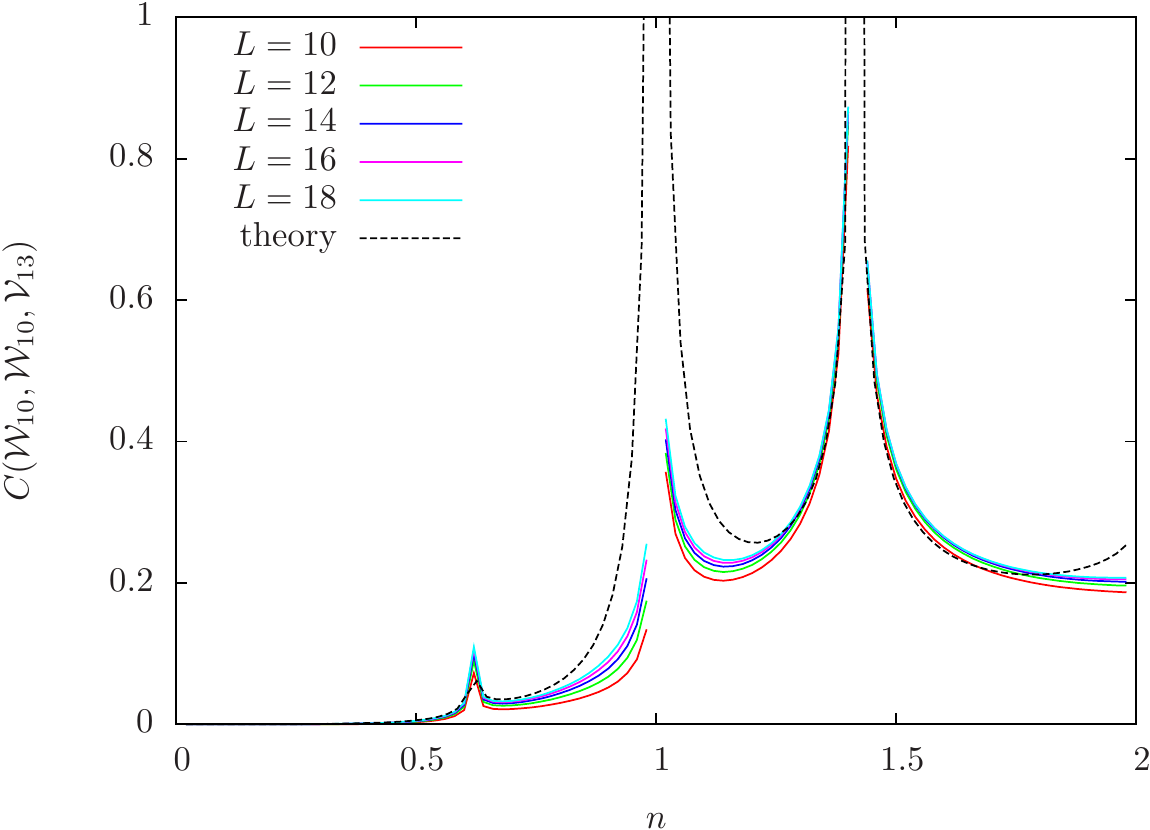}
  \end{center}
  \caption{Numerical data for the OPE coefficient $C(\W_{10},\W_{10},\V_{13})$, compared to the expression~\eqref{eq:C-theo1}.}
  \label{fig:C_w10_w10_v13}
\end{figure}

\begin{figure}
  \begin{center}
    \includegraphics{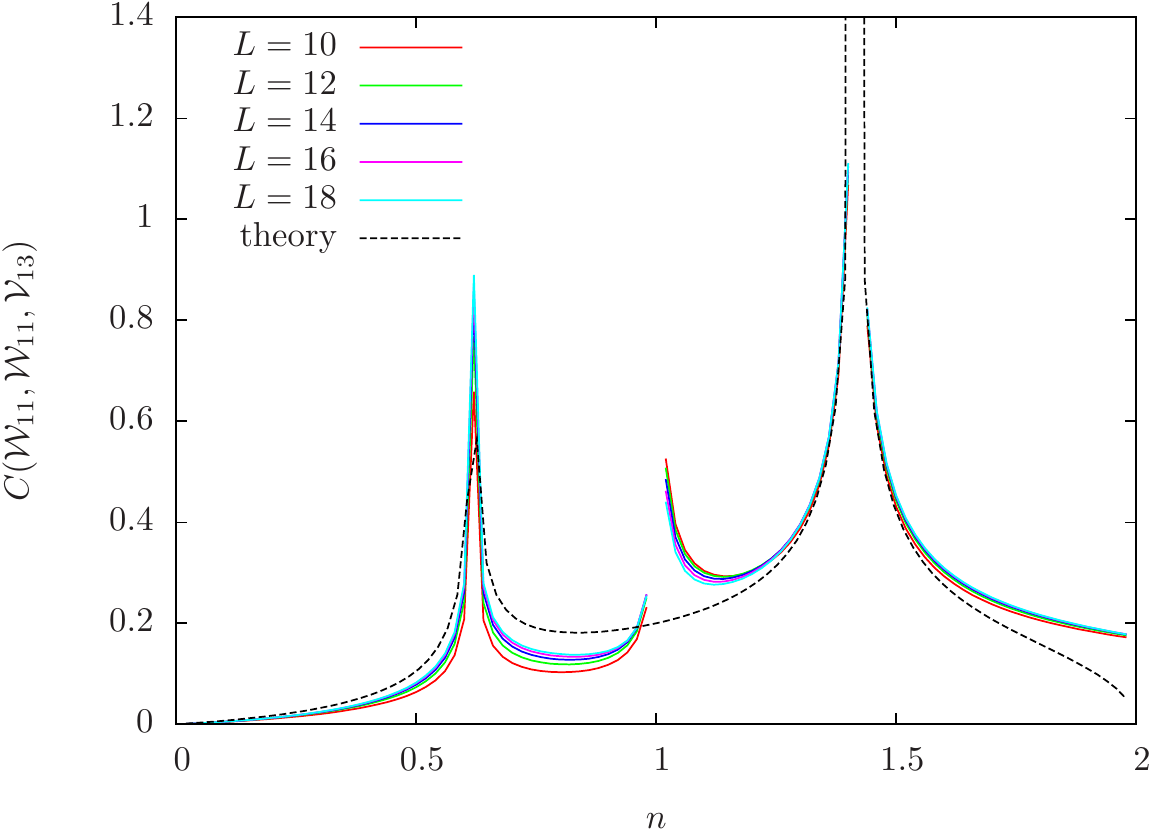}
  \end{center}
  \caption{Numerical data for the OPE coefficient $C(\W_{11},\W_{11},\V_{13})$, compared to the expression~\eqref{eq:C-theo2}.}
  \label{fig:C_w11_w11_v13}
\end{figure}

\begin{figure}
  \begin{center}
    \includegraphics{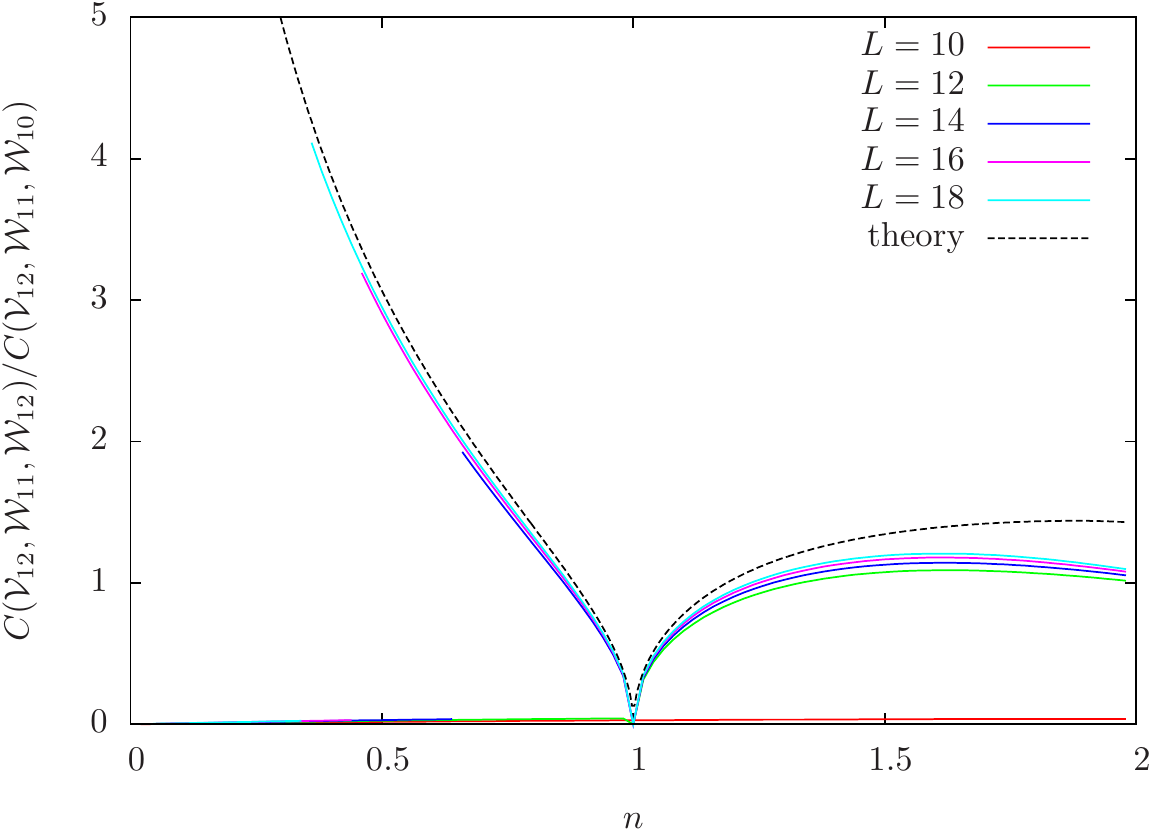}
  \end{center}
  \caption{Numerical data for the ratio $C(\V_{12},\W_{11},\W_{12})/C(\V_{12},\W_{11},\W_{10})$, compared to the expression~\eqref{eq:C-theo3}.}
  \label{fig:ratio_C_v12_w11_w12_C_v12_w11_w10}
\end{figure}

\begin{figure}
  \begin{center}
    \includegraphics{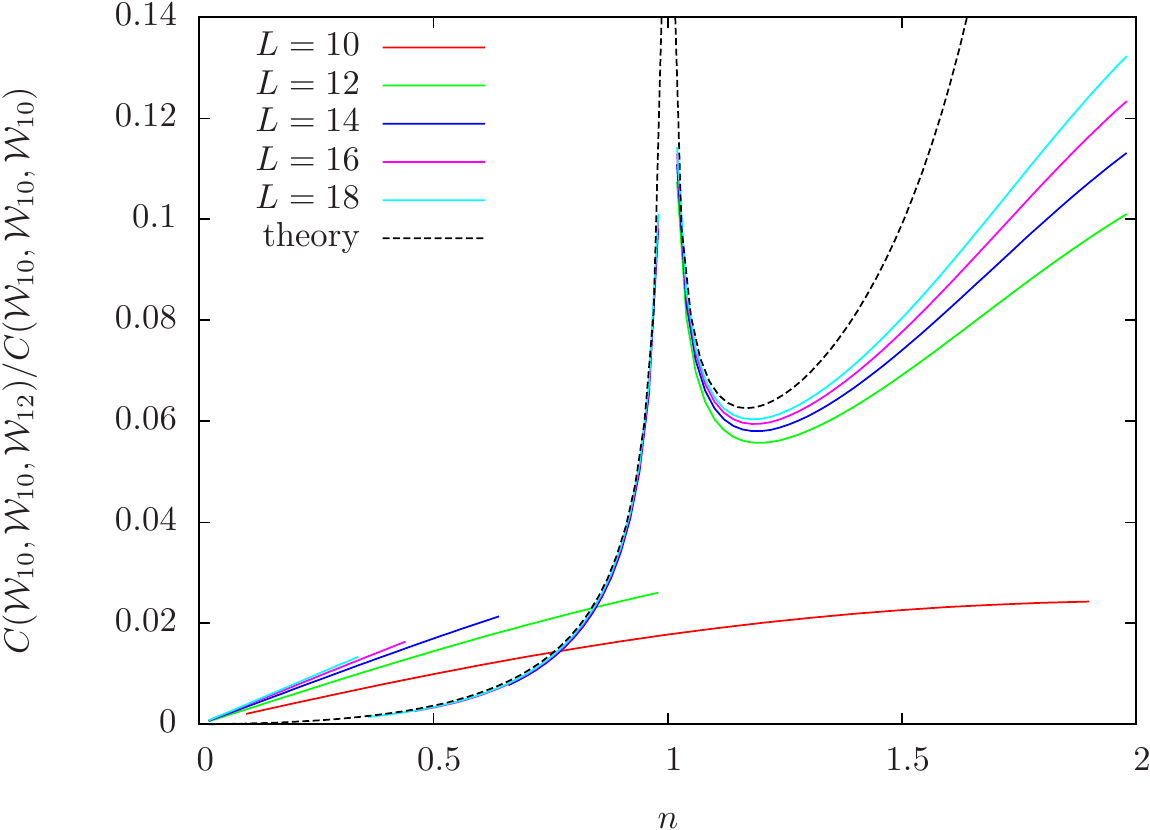}
  \end{center}
  \caption{Numerical data for the ratio $C(\W_{10},\W_{10},\W_{12})/C(\W_{10},\W_{10},\W_{10})$, compared to the expression~\eqref{eq:C-theo4}.}
  \label{fig:ratio_C_w10_w10_w12_C_w10_w10_w10}
\end{figure}

\section{Indecomposability in the bulk theory of generic loop models}
\label{sec:LCFT}

\subsection{The periodic Temperley-Lieb algebra}

\begin{figure}
  \begin{center}
    \includegraphics[scale=0.8]{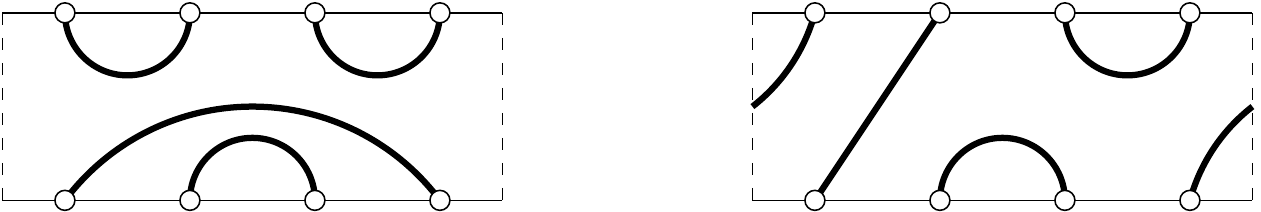}
  \end{center}
  \caption{Two distinct diagrams of $\PTL_{L=4}$ (see Definitions~\ref{def:diag}, \ref{def:PTL} and \ref{def:PTL2}). The top and bottom boundaries of the annulus are drawn as full lines, whereas periodic boundary conditions are assumed across the dotted lines. The parity of the left diagram is $\epsilon=1$, and the parity of the right diagram is $\epsilon=-1$ (see Definition~\ref{def:eps}).}
  \label{fig:diag}
\end{figure}

Following the discussion in~\cite{Azat14}, let us first describe the diagrams defining the algebra, and then give the algebraic rules between generators. We consider an annulus with $L$ marked points on the inner boundary and $L$ marked points on the outer boundary ($L$ is an even integer). Actually, we shall depict the annulus as a rectangle with periodic conditions across the vertical sides, and draw the inner (resp. outer) boundary at the bottom (resp. the top) of the rectangle: see \figref{diag}.

\begin{mydef} \label{def:diag}
  A {diagram} $w$ is a set of lines inside the annulus, connecting the $2L$ marked points with no intersections. The {arches} (resp. the legs) of $w$ are the lines connecting two points on the same boundary (resp. opposite boundaries).
\end{mydef}

We can consider diagrams as basis elements of an associative algebra:

\begin{mydef} \label{def:PTL}
  The periodic Temperley-Lieb algebra $\PTL_L$ is the algebra generated by the diagrams on an annulus with $2L$ marked points, where any two diagrams $(w_1,w_2)$ are identified (up to a multiplicative factor) according to the following rules:
  \begin{itemize}
  \item if $w_2$ can be obtained from $w_1$ by a continuous deformation of the lines, then $w_1=w_2$;
  \item if $w_2$ can be obtained from $w_1$ by a rotation of angle $2\pi$ of one of the boundaries, then $w_1=w_2$;
  \item if $w_2$ can be obtained from $w_1$ by removing a closed, contractible loop, then $w_1 = n\ w_2$;
  \item if $w_2$ can be obtained from $w_1$ by removing a closed, non-contractible loop, then $w_1 = \wt n\ w_2$;
  \end{itemize}
  and the product $w_1 w_2$ of two diagrams is obtained by gluing the bottom (or inner) boundary of $w_1$ onto the top (or outer) boundary of $w_2$.
\end{mydef}

Equivalently, we can formulate the above rules in terms of generators, which act by convention from bottom to top (or inner to outer) boundary:

\begin{mydef} \label{def:PTL2}
The periodic Temperley-Lieb algebra $\PTL_L$ is the algebra generated by $e_1, \dots e_L$ and the right-shift operator $\tau$, subject to the relations for $1 \leq j,k \leq L$:
\begin{equation}
  \begin{aligned}
    & e_j e_{j \pm 1} e_j = e_j \,, \\
    & e_j^2 = n \ e_j \,, \\
    & e_j e_k = e_k e_j \qquad \text{if} \quad 1<|j-k|<L-1 \,, \\
    & \tau^L = \id \,, \\
    & \tau e_j \tau^{-1} = e_{j+1} \,, \\
    & e_1 e_2 \dots e_{L-1} = \tau^2 e_{L-1} \,, \\
    & (e_2 e_4 \dots e_L) (e_1 e_3 \dots e_{L-1}) = \nt \ \tau (e_1 e_3 \dots e_{L-1}) \,,
  \end{aligned}
\end{equation}
where $e_{L+1} \equiv e_1$.
\end{mydef}

Finally, let us define a notion of parity:
\begin{mydef} \label{def:eps}
  The parity of a diagram $w$ is $\epsilon(w)=+1$ (resp. $\epsilon(w)=-1$) if the lines of $w$ intersect the vertical boundary an even (resp. odd) total number of times.
\end{mydef}
With this definition, we see readily that $\epsilon(w_1 w_2) = \epsilon(w_1)\epsilon(w_2)$.

\subsection{Representations of $\PTL_L$}

\begin{figure}
  \begin{center}
    \begin{tabular}{>{\centering\arraybackslash}m{0.5cm} >{\centering\arraybackslash}m{.7\textwidth}}
      $W_0$: & \includegraphics[scale=0.45]{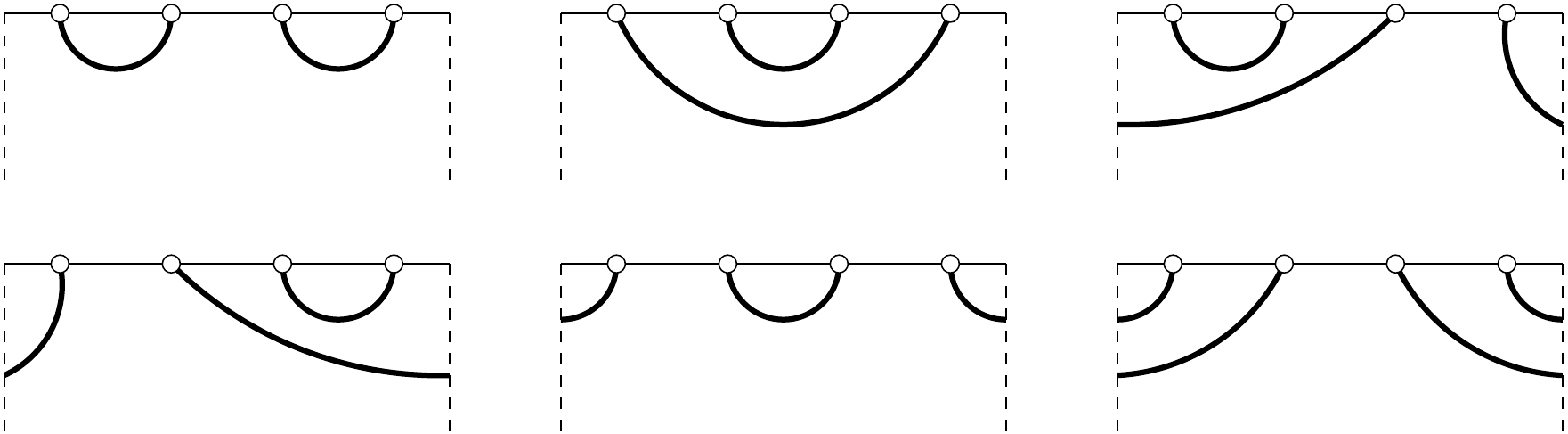}
      \\ \\ \\
      $W_1$: & \scalebox{0.45}{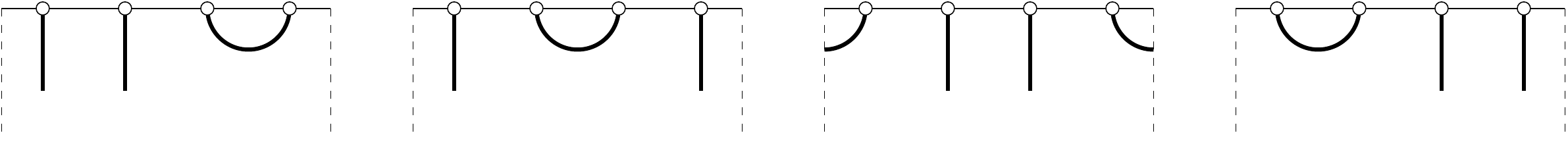}
      \\ \\ \\
      $W_2$: & \scalebox{0.45}{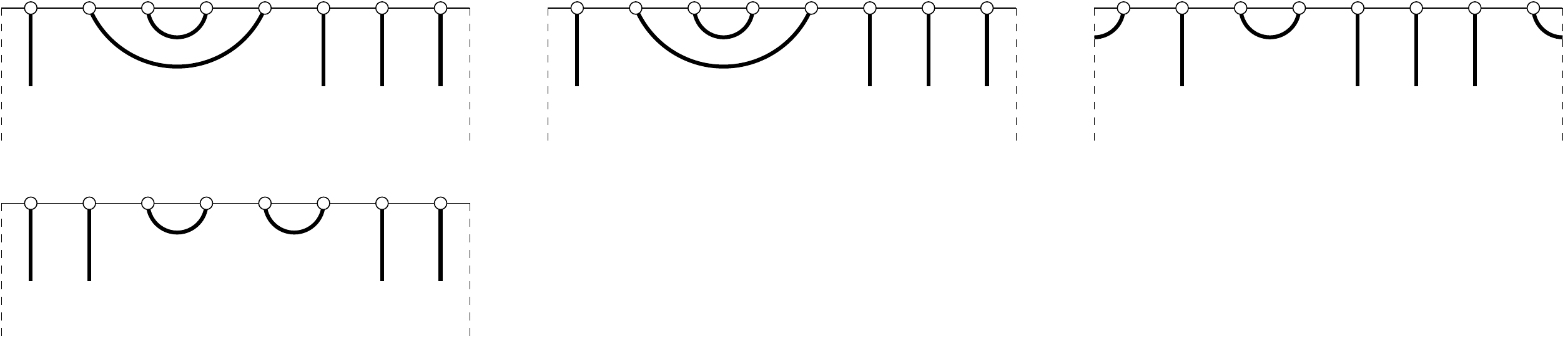}
    \end{tabular}
  \end{center}
  \caption{The basis states of $W_0$, $W_1$  for $L=4$ and $W_2$ for $L=8$.}
  \label{fig:state}
\end{figure}

\subsubsection{The spaces $W_\ell$ and the TL scalar product}
Under the action of the above algebra, the bottom arches are left unchanged, while the number of legs never increases. This suggests to define the action of $\PTL_L$ on the ``top-half'' of the diagram only. For this, we now consider a circle with $L$ marked points, with a cutline connecting the origin to the exterior of the circle. Equivalently, we can place the points on a horizontal segment, with periodic conditions across a vertical cutline.

\begin{mydef}
  A state $a$ (also denoted $\ket{a}$) is a set of non-intersecting lines drawn in the lower-half plane, with some lines (arches) connecting $(L-2\ell)$ marked points among themselves, and $2\ell$ vertical labelled lines (legs) attached to the remaining marked points, so that the labels form a cyclic permutation of $(1, 2, \dots 2\ell)$, and a leg attached to the point $j \in \{1, \dots, L\}$ carries a label with the same parity as $j$.
\end{mydef}
In particular, the top-half of a diagram $w$, consisting in the arches on the top boundary of $w$ together with the legs of $w$, defines a state. If we fix the number of legs, we get a representation of $\PTL_L$:
\begin{mydef}
  The representation $W_\ell$ is defined as the geometric action of $\PTL_L$ on the vector space generated by the states with $2\ell$ legs, where the contraction of any two legs gives zero. 
\end{mydef}

\begin{mydef}
The space $W=W_0 \oplus W_1 \oplus W_2 \oplus \dots$ is equipped with a bilinear, symmetric form $\aver{ \cdot\ , \ \cdot}$, defined as follows. Let $a$ and $b$ be two states with $\ell_a$ and $\ell_b$ legs, respectively. Consider the set of lines $g(a,b)$ obtained by gluing the reflection of $a$ around a horizontal axis with $b$, and set
\begin{equation}
  \aver{a,b} = \begin{cases}
    0 & \text{if any two legs of $a$ are joined in $g(a,b)$,} \\
    0 & \text{if any two legs of $b$ are joined in $g(a,b)$,} \\
    n^{N_{\rm c}(a,b)} \ \nt^{N_{\rm nc}(a,b)} & \text{otherwise,}
  \end{cases}
\end{equation}
where $N_{\rm c}(a,b)$ (resp. $N_{\rm nc}(a,b)$) is the number of contractible (resp. non-contractible) loops in $g(a,b)$. In particular, if $\ell_a \neq \ell_b$ then $\aver{a,b}=0$ automatically.
\end{mydef}

\begin{prop}
  \label{prop:self-adj}
  The generators $e_j$ are self-adjoint for the bilinear form $\aver{ \cdot\ , \ \cdot}$ on the space $W$:
  \begin{equation}
    \forall (a,b) \in W^2 \,,
    \qquad  \aver{e_j a, b} = \aver{a, e_j b} \,.
  \end{equation}
\end{prop}

\subsubsection{The map $\vphi$}
Let us now describe a linear map whose properties will allow us to relate the spaces $W_0$ and $W_1$.

\begin{mydef}
  Let $\vphi$ be the linear map defined as:
  \begin{equation} \label{eq:defphi}
    \vphi : \left\{\begin{array}{ccl}
        W_1 &\to& W_0 \\
        \ket{\dots \stackrel{j}{|} \dots \stackrel{k}{|} \dots}
        &\mapsto& \frac{(-1)}{2}^{j+1} \left(
          \ket{\dots \overset{j\dots k}{\cup} \dots} - \ket{\dots \overset{j\dots k}{\ncup} \dots}
        \right)
      \end{array} \right.
  \end{equation}
  where $\overset{j\dots k}{\cup}$ (resp. $\overset{j\dots k}{\ncup}$) means that the points $j$ and $k$ are connected by an arch which intersects the cutline an even (resp. odd) number of times.
\end{mydef}

We see immediately that
\begin{equation}
  e_j \vphi \ket{\dots \stackrel{j}{|} \ \stackrel{j+1}{|} \dots} = \frac{(-1)}{2}^{j+1} (n-\nt)
  \ket{\dots \overset{\ \ j \ j+1}{\cup} \dots}\,.
\end{equation}
Hence, taking all the signs into account, $\vphi$ is a morphism of algebras, up to a parity sign:
\begin{prop}
  \label{prop:phi-morph}
  If $n=\nt$, then
  \begin{equation}
    \forall w \,, \quad \forall a \in W_{1} \,,
  \qquad w\ \vphi \ket{a} = \epsilon(w) \ \vphi w \ket{a} \,.
\end{equation}
\end{prop}
\begin{prop}
  \label{prop:phi-null-state}
  If $n=\nt$, then
  \begin{equation}
    \forall a \in W_{0} \,,
    \quad \forall b \in W_{1} \,,
    \qquad \aver{a, \vphi(b)} = 0 \,.
  \end{equation}
\end{prop}
In particular, $\vphi(b)$ has zero norm: $\aver{\vphi(b), \vphi(b)}=0$.

Note that a state with no legs can be specified by a planar pairing of the $L$ marked points, with each arch decorated by the parity of its number of intersections with the cutline. In the case $n=\nt$, one may define a representation which ignores the parities:
\begin{mydef}
  A zero-leg reduced state $\ket{\wh a}$ is a planar pairing of the $L$ marked points. When $n=\nt$, the representation $\wh W_0$ is defined as the geometric action of $\PTL_L$ on the vector space generated by the zero-leg reduced states.
\end{mydef}
With some simple counting, one gets the dimensions of the above representations:
\begin{equation} \label{eq:dimW}
  \dim W_\ell = \bino{L}{L/2-\ell} \,,
  \qquad \dim \wh W_0 = \frac{L!}{(L/2)!(L/2+1)!} \,.
\end{equation}
\begin{prop}
  \label{prop:phi-inj}
  If $n=\nt$, then $\vphi$ is injective, and
  \begin{equation} \label{eq:direct-sum}
    W_{0} = \wh W_{0} \oplus \Im(\vphi) \,,
  \qquad \Im(\vphi) \cong W_{1} \,.
  \end{equation}
\end{prop}
\begin{proof}
  Consider a state $a \in W_{0}$. If $a$ has no arch with odd parity, then $a \in \wh W_{0}$. Otherwise, we consider the most external arch of $a$ with odd parity, and denote by $j<k$ the ends of this arch. We define
  \begin{equation}
    \ket{a} = \ket{\dots \overset{j\dots k}{\ncup} \dots} \,,
    \quad \ket{a'} = \ket{\dots \overset{j\dots k}{\cup} \dots} \,,
    \quad \ket{b} = 2(-1)^j \ \ket{\dots \stackrel{j}{|} \dots \stackrel{k}{|} \dots} \,.
  \end{equation}
  From~\eqref{eq:defphi}, we have
  \begin{equation}
    \ket{a} = \ket{a'} + \vphi\ket{b} \,,
  \end{equation}
  where $a$ has one less arch with odd parity than $a$. By iterating this process,
  we can write $a$ as
  \begin{equation}
    \ket{a} = \ket{a'_k} + \vphi\ket{b_k} \,,
    \qquad
    a'_k \in \wh W_0 \,,
    \quad
    b_k \in W_1 \,.
  \end{equation}
  Hence, we see that
  \begin{equation}
    W_0 \subseteq \wh W_0 + \Im(\vphi) \quad\Rightarrow\quad
    \dim W_0 \leq \dim \wh W_0 + {\rm rk}(\vphi) \,.
  \end{equation}
  But, from~\eqref{eq:dimW}, we have
  \begin{equation}
    \dim W_0 = \dim \wh W_0 + \dim W_1 \,,
  \end{equation}
  which leads to ${\rm rk}(\vphi) \geq \dim W_1$, and proves that
  $\vphi$ is injective.
\end{proof}

\subsubsection{Degeneracies of $H$ between the $W_0$ and $W_1$ sectors}

Let us consider, for $n=\nt$, the Hamiltonian associated to the TL model:
\begin{equation} \label{eq:ham}
  H = -\sum_{j=1}^L e_j \,.
\end{equation}

\begin{prop} \label{prop:deg}
  To any eigenstate $\ket\W \in W_1$ of $H$, one can associate a degenerate eigenstate $\ket\chi = \vphi \ket\W \in W_0$, and $\ket \chi$ is orthogonal to $W$. In particular, $\aver{\chi,\chi}=0$.
\end{prop}

\begin{proof}
  Consider an eigenstate $\ket\W \in W_1$, with energy $E$:
  \begin{equation}
    H \ket\W = E \ket\W \,.
  \end{equation}
  From \propref{phi-inj}, the state $\ket\chi=\vphi\ket\W \in W_0$ is non-zero. Moreover, since every term in $H$ has even parity, we have from~\propref{phi-morph}:
  \begin{equation}
    H \ket\chi = H \vphi\ket\W = \vphi H \ket\W = E \ket\chi \,,
  \end{equation}
  and hence $\ket\chi$ is also an eigenstate of $H$, with eigenvalue $E$. Note that, from~\propref{phi-null-state}, $\ket\chi$ is orthogonal to the space $W_0$, and therefore it is orthogonal to every state in $W$: we say that $\ket\chi$ is a decoupling state.
\end{proof}

From the expression~\eqref{eq:h} of conformal dimensions, which determine the energies in the continuum limit, we see that, for generic $n$, the above are the only possible degeneracies between two sectors $W_{\ell}$ and $W_{\ell'}$ with $\ell \neq \ell'$. Note that the situation is very different when $g$ is rational (see {\it e.g.} \cite{Azat14}).

\subsubsection{Mixed representation}

Some of the correlation functions in the loop model cannot be expressed in terms of scalar products in $W$, but instead they impose the use of a representation in which the contraction of two legs is allowed.

\begin{mydef} The representation $W_{\leq 1}$ is defined as the geometric action of $\PTL_L$ on the space $W_0 \oplus W_1$, with the following coefficients for contractions:
  \begin{equation} \label{def-ej}
    \begin{aligned}
      &e_j \ket{\dots \overset{j}{|}\ \overset{j+1}{|}\dots}
      = \epsilon_j \ket{\dots \overset{j\ j+1}{\cup} \dots} \,,  \\
      &e_L \ket{\overset{1}{|} \dots \overset{L}{|}}
      = \epsilon_L \ket{\overset{1 \dots L}{\ncup}} \,,
    \end{aligned}
  \end{equation}
  with $\epsilon_j = (-1)^{j+1}$.
\end{mydef}

\begin{mydef}
  The space $W_{\leq 1}$ is equipped with a symmetric, bilinear form $\Aver{\cdot\ ,\ \cdot}$, defined as
  \begin{equation}
    \Aver{a,b} =
    \sigma(a,b)\ \sigma'(a,b)\ n^{N_{\rm c}(a,b)} \ \nt^{N_{\rm nc}(a,b)} \,,
  \end{equation}
  where the notations are the same as in the definition of $\aver{a,b}$, and
  $$
  \sigma(a,b) = \begin{cases}
    (-1)^{p_{ab}+j_a+1} & \text{\begin{tabular}{l}
        if two legs of $a$, connected to the points $j_a<k_a$ \\
        are joined in $g(a,b)$ by crossing the cutline $p_{ab}$ times,
      \end{tabular}} \\
    1^{\phantom{\Big|}} & \text{otherwise,}
  \end{cases}
  $$
  and $\sigma'(a,b)$ is defined similarly in terms of the legs of $b$.
\end{mydef}

\begin{prop}
  The generators $e_j$ are self-adjoint for the bilinear form $\Aver{\cdot\ ,\ \cdot}$ on the space $W_{\leq 1}$:
  \begin{equation}
    \forall (a,b) \in W_{\leq 1}^2 \,,
    \qquad \Aver{e_j a \,, b} = \Aver{a \,, e_j b} \,.
  \end{equation}
\end{prop}

\subsubsection{Jordan blocks}
\label{sec:jordan}

Let us set $n=\nt$.
For any pair of degenerate eigenstates $\ket\W$ and $\ket\chi=\vphi\ket\W$ discussed in~\propref{deg}, we shall argue that $H$ has a Jordan block in the representation $W_{\leq 1}$, with matrix:
\begin{equation} \label{eq:jordan}
  \left( \begin{array}{cc} 
      E & 1 \\
      0 & E
    \end{array} \right) \,.
\end{equation}
Let us view $\ket\W$ and $\ket\chi$ as vectors in the representation $W_{\leq 1}$. The action of $H$ on $\ket\chi$ is unchanged, since $\ket\chi$ is in the zero-leg sector. For $\ket\W$, we have to include the zero-leg terms arising after a contraction, which we denote $\ket{\W_c}$:
\begin{equation}
  H \ket\W = E \ket \W + \ket{\W_c} \,.
\end{equation}
Suppose there exists a solution $\ket{\W'}$ of the linear system in $W_0$:
$(H - E) \ket{\W'} = \ket\chi - \ket{\W_c}$. We have checked numerically the existence of a solution\footnote{Note that the signs $\epsilon_j$ in the action of $H$ on $W_{\leq 1}$ are crucial, since the momenta of $\ket\W$ and $\ket\chi=\vphi\ket\W$ differ by $\pi$. If instead we had set $\epsilon_j=1$ (and $\sigma(a,b)=\sigma'(a,b)=1$ in the corresponding scalar product so that $H$ is still Hermitian), then $H$ would have a diagonal form in the block $(\Ket\chi,\Ket\W)$.} in many cases, including all the pairs of degenerate states mentioned in this paper. Note that for any scalar $\lambda$, the change $\ket{\W'} \to \ket{\W'} + \lambda \ket\chi$ does not affect the above linear system, and thus the solution is not unique. We introduce the states
\begin{equation}
  \Ket\chi = \ket\chi \,,
  \qquad \Ket\W = \ket\W + \ket{\W'} \,,
\end{equation}
and we get
\begin{equation}
  H \Ket\chi = E \Ket\chi \,,
  \qquad
  H \Ket\W = E \Ket\W + \Ket\chi \,.
\end{equation}

In summary, given an eigenstate $\ket\W$ with energy $E$ in the two-leg representation $W_1$, we have constructed a pair of states $(\Ket\chi, \Ket\W)$ in the representation $W_{\leq 1}$ so that $H$ has the Jordan form~\eqref{eq:jordan} in ${\rm span}(\Ket\chi, \Ket\W)$, and $\Ket\chi$ is a state in the zero-leg sector of $W_{\leq 1}$ with vanishing norm: $\Aver{\chi,\chi}=0$.

\subsection{Scaling limit}

Throughout this section, we set $n=\nt$.

\subsubsection{Identification of states}

The states $\{ \ket{\V_{1k}} \}$ and $\{ \ket{\W_{me}} \}$ defined by vertex operators in~\secref{vertex} can be identified as the scaling analogs of states in the representations $\{W_\ell \}$. The general idea is that the continuous analog of $\PTL_L$ is made of two copies of the Virasoro algebra: $\Vir \oplus \ol{\Vir}$, and that the discrete quantum numbers of $H$ survive in the scaling limit.

For clarity, we denote $\ket{\Phi}_L$ and $\Ket{\Phi}_L$ a state in $W_\ell$ and $W_{\leq 1}$ respectively, and $\Ket{\Phi}$ the corresponding state in the CFT Hilbert space. We use a similar notation for scalar products.

\begin{itemize}

\item The state $\ket{\V_{1k}}_L$, with $k \geq 1$, belongs to $\wh W_0 \subset W_0$, and it has finite norm (which is set to one by convention), since the scalar product $\aver{\cdot,\cdot}$ is non-degenerate in $\wh W_0$ for generic $n$.

\item The state $\ket{\V_{1k}}_L$, with $k \leq 0$, belongs to $\Im \vphi \subset W_0$, and it totally decouples:
\begin{equation}
  \forall k \leq 0 \,, \quad \forall a \in W_0 \,, \qquad \aver{\V_{1k}, a}_L = 0 \,.
\end{equation}

\item The state $\ket{ \W_{me} }_L$, with $m \geq 1$ and $e \in \Zb$, belongs to $W_m$. 
Recall that the allowed electric charges are of the form $e=p/m$, with $p \in \Zb$, which may
be interpreted as follows. Consider the cyclic translation $R$ of the leg labels by two units:
\begin{equation}
  R \ket{\dots \underset{x_1}{|} \dots \underset{x_2}{|} \dots \underset{x_3}{|} \dots}
   = \ket{\dots \underset{x_1+2}{|} \dots \underset{x_2+2}{|} \dots \underset{x_3+2}{|} \dots} \,,
\end{equation}
where the labels are considered {modulo} $2m$. The action of $H$ on $W_m$ commutes with $R$, and we can label the eigenstates of $H$ according to the eigenvalues $r_q=\exp(2i\pi q/m)$ of $R$, with $q=0, \dots, m-1$. Then the state $\ket{ \W_{me} }_L$ with $e=p/m$ is in the sector $r_q$ such that $p \equiv q \mod m$. Again, since the scalar product $\aver{\cdot,\cdot}$ is non-degenerate in $W_m$, one can set $\aver{\W_{me}, \W_{me}}_L = 1$.

\item For $m=1$ and $k \geq 1$, $\W_{1k}$ is degenerate with the null descendant of $\V_{1k}$ (under the action of $\ol\Vir$) at level $k$, which we denote $\bar\chi_{1k}$. In the representation $W_0 \oplus W_1,$ the two corresponding states $(\ket{\bar\chi_{1k}}_L, \ket{\W_{1k}}_L)$ are related by
\begin{equation}
  \ket{\bar\chi_{1k}}_L = \vphi \ket{\W_{1k}}_L \,,
\end{equation}
and they are independent, orthogonal eigenstates. In contrast, in $W_{\leq 1}$, the corresponding states form a Jordan block, and we use a disctinct notation $(\Ket{\bar\chi_{1k}}_L, \Ket{\W_{1k}}_L)$ for them. The Jordan block~\eqref{eq:jordan} becomes, in the scaling limit:
\begin{equation} \label{eq:jordan2}
  L_0 + \bar L_0 = \left(\begin{array}{cc}
      h_{1k}+h_{1,-k} & 1 \\
      0 & h_{1k}+h_{1,-k}
    \end{array} \right) \,.
\end{equation}
There is a similar relation between $\W_{1,-k}$ and $\chi_{1k}$, the null descendant of $\V_{1k}$ under the action of $\Vir$.

Moreover, as we have argued above, any eigenstate in $W_1$ produces a $2 \times 2$ Jordan block in $W_{\leq 1}$. In particular, since $\vphi$ commutes with the action of $\PTL_L$ (which generates $\Vir \oplus \ol\Vir$ in the scaling limit), we expect that in the scaling limit
\begin{equation}
  \vphi L_{-n_1} \dots L_{-n_p} \ket{\W_{1k}} = L_{-n_1} \dots L_{-n_p} \vphi \ket{\W_{1k}}
  = L_{-n_1} \dots L_{-n_p} \ket{\bar\chi_{1k}}
\end{equation}
and thus $L_{-n_1} \dots L_{-n_p} \Ket{\W_{1k}}$ and $L_{-n_1} \dots L_{-n_p} \Ket{\bar\chi_{1k}}$ will also form a Jordan block.

\item Finally, let us comment on the case of $\W_{10}$. Here, the vertex operator $\V_{10}$ is a null state, and in $W_{\leq 1}$ we have a Jordan block for the pair $(\Ket{\V_{10}}_L, \Ket{\W_{10}}_L)$.
\end{itemize}

\subsubsection{Hilbert space}

A crucial question is then to decide which representations should be included in the Hilbert space, so that the scaling limit results in a consistent CFT, {\it i.e.} with a consistent operator algebra. Indeed, several unequivalent representations correspond to the same spectrum in the continuum limit.

Let us first examine the choice $W=W_0 \oplus W_1 \oplus W_2 \dots$ as the Hilbert space, with the associated scalar product $\aver{\cdot, \cdot}$. In this representation, the zero-norm vectors span the subspace $\Im(\vphi) \subset W_0$, and they actually decouple from the whole vector space $W$. Once they have been removed, we are left with $\wh W= \wh W_0 \oplus W_1 \oplus W_2 \dots$, where the scalar product $\aver{\cdot, \cdot}$ is non-degenerate for generic $n$. Hence, in the corresponding CFT in the scaling limit, the 2- and 3-point functions will have the standard form, and the ordinary differential equations arising from null-vector conditions will hold. These two facts are sufficient to ensure fusion rules of the form
\begin{equation}
  \Phi_{1,m} \times \Phi_{1,n} \to \Phi_{1,|n-m|+1} + \dots + \Phi_{1,n+m-1} \,,
\end{equation}
separately in the holomorphic and anti-holomorphic sectors. However, the results exposed in the end of \secref{funct} show that these rules are not obeyed, e.g. the fusion $\W_{12} \times \W_{10} \to \W_{10}$ is allowed.

This contradiction is resolved if we use instead:
\begin{equation} \label{eq:hilbert}
  \boxmath{\mathcal{H} = W_{\leq 1} \oplus W_2 \oplus W_3 \oplus \dots}
\end{equation}
with the associated scalar product $\Aver{\cdot,\cdot}$. In simple terms, it means we take the same basis of states as in $W$, but we allow the mixing of some states from the zero- and two-leg sectors, both in the action of the algebra and in the scalar product.

In the following, we shall assume that the representation~\eqref{eq:hilbert} is the one that produces a consistent operator algebra in the scaling limit.

\subsubsection{Decoupling states}
\label{sec:decoupling}

We are now ready to discuss which of the zero-norm states actually decouple from the whole Hilbert space $\mathcal{H}$ under the scalar product $\Aver{\cdot, \cdot}$. Let us denote by $A_{1k} \in \Vir$ the combination of generators which produces a null descendant at level $k$ in the module of $\V_{1k}$, and $\bar A_{1k}$ the analogous object in $\ol\Vir$:
\begin{equation}
  \ket{\chi_{1k}} = A_{1k} \ket{\V_{1k}} \,,
  \qquad
  \ket{\bar\chi_{1k}} = \bar A_{1k} \ket{\V_{1k}} \,.
\end{equation}
As explained above, the zero-norm state $\ket{\chi_{1k}}$ (resp. $\ket{\bar\chi_{1k}}$) forms a Jordan block with $\ket{\W_{1,-k}}$ (resp. $\ket{\W_{1k}}$). However, this fact does not tell us anything about the scalar products $\aver{\chi_{1k},\W_{1,-k}}=\aver{\bar\chi_{1k},\W_{1k}}$, and at present we are not in a position to make any statement on the decoupling of the vectors $\ket{\chi_{1k}}$ and $\ket{\bar\chi_{1k}}$. We did check numerically that the scalar products $\Aver{\bar\chi_{1k}, \W_{1k}}_L  = \Aver{\chi_{1k}, \W_{1,-k}}_L$ for $k=0,1,2$ (setting $\chi_{10}=\bar\chi_{10} \equiv \V_{10}$) do not vanish at finite size $L=2, \dots, 12$, but considerable additional work is needed to conclude about the scaling limit of these scalar products.

The operator $\W_{1k}$ (resp. $\W_{1,-k}$) also admits a null descendant $A_{1k} \W_{1k}$ (resp. $\bar A_{1k} \W_{1,-k}$), since it has conformal dimensions $(h_{1k},h_{1,-k})$ [resp. $(h_{1,-k},h_{1,k})$]. Acting with $\vphi$ in the scaling limit, we get
\begin{equation}
  \vphi(A_{1k} \ket{\W_{1k}}) = \vphi(\bar A_{1k} \ket{\W_{1,-k}}) = A_{1k} \bar A_{1k} \ket{\V_{1k}} \,,
\end{equation}
and thus $A_{1k} \ket{\W_{1k}} = \bar A_{1k} \ket{\W_{1,-k}}$, since $\vphi$ is injective. Indeed, in the lattice model, we find a unique state in $W_1$ with energy $\sim \frac{2\pi}{L} \times 2h_{1,-k}$ and zero momentum, degenerate with a unique state in $W_0$. In the scaling limit, the scalar product between the corresponding states is $\Aver{A_{1k}\bar A_{1k} \V_{1k}, A_{1k} \W_{1k}} = 0$, because $A_{1k}^\dag A_{1k}^{\phantom\dag} \Ket{\W_{1k}} = 0$ by construction. Hence, the states $A_{1k} \Ket{\W_{1k}}=\bar A_{1k} \ket{\W_{1,-k}}$ and $A_{1k} \bar A_{1k} \Ket{\V_{1k}}$ are totally decoupled, and we may set them to zero in the Hilbert space of the CFT:
\begin{equation}
  A_{1k} \Ket{\W_{1k}} = \bar A_{1k} \Ket{\W_{1,-k}} \equiv 0 \,,
  \qquad \text{and} \qquad
  A_{1k} \bar A_{1k} \Ket{\V_{1k}} \equiv 0 \,.
\end{equation}
Note that, in the case $k=1$, the lattice analogs of $W_{1, \pm 1}$ are two of the ``discrete parafermions''~\cite{IC09,IWWZ} of the TL or $\On$ model, i.e. lattice operators satisfying a discrete version of the Cauchy-Riemann equations $\partial W_{1,1} = 0$ and $\bar\partial W_{1,-1}=0$.

\subsection{OPEs in the presence of Jordan blocks}

\subsubsection{Limiting procedure}

The Jordan block structures discussed in previous sections appear at the point $n=\wt n$, for generic $n$. We can thus follow a treatment of the singular OPEs similar to the $c \to 0$ limit of~\cite{Gurarie93,Cardy01}, except that in our case we have a two-parameter family of CFTs (labelled by $g$ and $e_0$ in the notations of~\secref{model}): we shall fix the coupling constant $g$, and let the background charge $e_0$ tend to $(1-g)$.

More specifically, we have seen that for general parameters $(n,\nt)$, we get for the TL loop model the central charge and dimensions~\eqref{eq:h} of the form:
\begin{equation}
  c(g,e_0)= 1 - \frac{6e_0^2}{g} \,,
  \qquad
  h_{me}(g,e_0) = \frac{1}{4}\left(m\sqrt{g} - \frac{e}{\sqrt g} \right)^2
  - \frac{e_0^2}{4g} \,,
\end{equation}
where $m=0, e \in e_0+\Zb$ for electric operators $\V_{0e}$, and $m \in \{1,2,3 \dots\}, e \in \Zb/m$ for electro-magnetic operators $\W_{me}$, and
\begin{equation}
  n = -2 \cos \pi g \,, \qquad \nt = 2\cos \pi e_0 \,.
\end{equation}
In contrast, the dimensions of the Kac table for central charge $c(g,e_0)$ are of the form
\begin{equation}
  h^{(K)}_{rs}(g,e_0) = \left(r\alpha_- + s \alpha_+ \right)^2
  - \frac{e_0^2}{4g} \,,
\end{equation}
with $(r,s)$ integers, and
\begin{equation}
  \alpha_+ + \alpha_- = -\frac{e_0}{\sqrt{g}} \,,
  \qquad
  \alpha_+ \alpha_- = -1 \,.
\end{equation}

The general idea which we develop in the following is to consider an OPE of two operators, such that one of the terms in the expansion is a null state belonging to a Jordan block of $L_0+\bar L_0$ for $n=\nt$, and analyse the behaviour of OPE coefficients in the limit $\nt \to n$.

In particular, the following dimensions of vertex operators coincide with the first row of the Kac table in this limit:
\begin{equation}
  h_{0,e_0+k} \to h^{(K)}_{1,k+1} \,,
  \qquad
  h_{1,k+1} \to h^{(K)}_{1,k+1} \,.
\end{equation}

\subsubsection{OPE coefficients}

Consider the OPE between two primary fields $\Phi_1$ and $\Phi_2$, with dimensions $(h_1,\hb_1)$ and $(h_2,\hb_2)$:
\begin{equation} \label{eq:OPE}
  \Phi_1(z,\zb) \Phi_2(0) = \sum_p C(\Phi_1, \Phi_2, \Phi_p)\ z^{-h_1-h_2+h_p} \zb^{-\hb_1-\hb_2+\hb_p} \left[\Phi_p(0) + \dots \right] \,,
\end{equation}
where the sum is over all the primary fields $\Phi_p$ appearing in the fusion  $\Phi_1 \times \Phi_2$, and the $\dots$ denote the contributions from the descendants of $\Phi_p$. In the following, we shall need the precise expression of the OPE coefficients of some of these descendants. By the standard methods of CFT, we obtain
\begin{equation}
  C(\Phi_1, \Phi_2, L_{-k_1} \dots L_{-k_m} \Phi_p)
  = C(\Phi_1, \Phi_2, \Phi_p) \times \beta_{12}^{p(k_1,\dots k_m)} \,,
\end{equation}
where the $\beta_{12}^{p(k_1,\dots k_m)}$ satisfy a set of linear recursion identities. At level one we have
\begin{equation} \label{eq:beta1}
  2h_p \beta_{12}^{p(1)} = h_1-h_2+h_p \,.
\end{equation}
At level two, the coefficients of $L_{-1}^2\Phi_p$ and $L_{-2}\Phi_p$ are determined by the system
\begin{equation}
  \begin{aligned}
    &2(2h_p+1) \beta_{12}^{p(1,1)} + 3 \beta_{12}^{p(2)} = (h_1-h_2+h_p+1) \beta_{12}^{p(1)} \,, \\
    &6h_p \beta_{12}^{p(1,1)} + (4h_p+c/2) \beta_{12}^{p(2)} = (2h_1-h_2+h_p) \,,
  \end{aligned}
\end{equation}
whose solution may be written as
\begin{equation}
  \begin{aligned}
    &\beta_{12}^{p(2)} = \frac{P(h_1,h_2,h_p)}{16h_p^2+2(c-5)h_p+c} \,, \\
    &\beta_{12}^{p(1,1)} = \frac{2(2h_1-h_2+h_p)-(8h_p+c) \beta_{12}^{p(2)}}{12h_p} \,,
  \end{aligned}
\end{equation}
where
\begin{equation}
  P(h_1,h_2,h_p) = -3(h_1-h_2+h_p)(h_1-h_2+h_p+1) + 2(2h_p+1)(2h_1-h_2+h_p) \,.
\end{equation}

\subsubsection{Null states in the OPE}

As we have seen in previous sections, at $n=\nt$, the null descendant $\chi_{1k}$ at level $k$ of the primary operator $\V_{1k}$ may survive in the theory, forming a Jordan block in $L_0+\bar L_0$ with its logarithmic partner $\W_{1,-k}$. Let us analyse in detail the case of $\V_{12}$, which is the energy operator in the TL loop model.

Consider two primary fields $\Phi_1$ and $\Phi_2$, such that their fusion contains the terms
\begin{equation}
  \Phi_1 \times \Phi_2 \to \V_{0,e_0+1} + \W_{1,-2} + \dots
\end{equation}
where the $\dots$ denote primary fields with distinct dimensions. In the limit $\nt \to n$, we have
\begin{equation}
  h = h_{0,e_0+1} \to h^{(K)}_{12} \,,
  \qquad
  h' = h_{12} \to h^{(K)}_{12} \,.
\end{equation}
The contributions from level-two descendants of $\V_h=\V_{0,e_0+1}$ in the OPE~\eqref{eq:OPE} can be organised as
\begin{equation} \label{eq:OPE-term1}
  \frac{P(h_1,h_2,h)}{16h^2+2(c-5)h+c}
  \left( L_{-2} - \frac{8h+c}{12h} L_{-1}^2 \right) \V_h
  + \frac{2h_1-h_2+h}{6h} L_{-1}^2 \V_h \,.
\end{equation}
The numerator $P(h_1,h_2,h)$ of the first term is a polynomial which vanishes if $h=h^{(K)}_{12}$ and $(h_1,h_2)$ obey the standard fusion rule with $\Phi_{12}$ (expressed in vertex charges $h_j=\alpha_j^2 - 2\alpha_0\alpha_j$) $\alpha_2 = \alpha_1 \pm \alpha_+/2$, whereas the denominator has roots $h_{12}^{(K)}$ and $h_{21}^{(K)}$.

If the dimensions $h_1$ and $h_2$ obey the standard fusion rule with $\Phi_{12}$ in the limit $\nt \to n$, then all the coefficients in~\eqref{eq:OPE-term1} remain finite. As $\nt \to n$, the first term converges to the null state
\begin{equation}
  \left(L_{-2} - g L_{-1}^2 \right) \V_{12} \equiv \chi_{12} \,.
\end{equation}

On the other hand, more interestingly, we shall now show that the presence of the operator $\W_{1,-2}$ allows the fusion $\Phi_1 \times \Phi_2 \to \V_{12}$, even if $\alpha_2 \neq \alpha_1 \pm \alpha_+/2$. In this case, we define the finite quantity
\begin{equation}
  Q_{12} = \frac{P(h_1,h_2,h)}{16(h - h^{(K)}_{21})} \,,
\end{equation}
and~\eqref{eq:OPE-term1} takes the form
\begin{equation} \label{eq:OPE-term1-lim}
  \frac{Q_{12}}{h - h^{(K)}_{12}}\
  (L_{-2} - g L_{-1}^2) \V_h + \#\ L_{-1}^2 \V_h \,,
\end{equation}
where $\#$ denotes a finite coefficient. Hence, in the limit $\nt\to n$, the null state $\chi_{12}$ appears in the OPE of $\Phi_1 \times \Phi_2$ with a diverging coefficient. This diverging term in the OPE~\eqref{eq:OPE} should be cancelled by a term coming from $\W_{1,-2}$, which tells us that the latter operator should decompose as
\begin{equation} \label{eq:partner}
  \W_{1,-2} = \frac{C(\Phi_1,\Phi_2,\V_h)}{C(\Phi_1,\Phi_2,\W_{1,-2})}
  \times \frac{Q_{12}}{h - h^{(K)}_{12}}\
  \left[ -(L_{-2} - g L_{-1}^2) \V_h + \kappa \ w_{1,-2} \right] \,,
\end{equation}
where $w_{1,-2}$ has a finite limit as $\nt\to n$, and $\kappa$ is a vanishing coefficient. Note that $w_{1,-2}$ is not defined uniquely, since it we can add a term $\propto (L_{-2} - g L_{-1}^2) \V_h$ to it.

We expand the power factors which multiply the term~\eqref{eq:OPE-term1-lim} in the OPE:
\begin{align}
  &z^{-h_1-h_2+h+2} \zb^{-h_1-h_2+h} = \nn \\
  &\qquad z^{-h_1-h_2+h'+2} \zb^{-h_1-h_2+h'} \times \left[
  1 + (h - h') \log(z\zb) + \mathrm{O}(\nt-n) 
  \right] \,.
\end{align}
Then, if we set the normalising factor in~\eqref{eq:partner} to $\kappa=h-h'$, the OPE~\eqref{eq:OPE} takes the well-defined form
\begin{equation}
\boxmath{\begin{aligned}
  \Phi_1(z,\zb) \Phi_2(0) =& z^{-h_1-h_2+h_{12}+2} \zb^{-h_1-h_2+h_{12}}
  C(\Phi_1, \Phi_2, w_{1,-2}) \\
  & \times \left\{ w_{1,-2}(0) + [(L_{-2} - g L_{-1}^2)\V_{0,e_0+1}](0) \log(z\zb)
  \right\} +\dots
\end{aligned}}
\end{equation}
where the $\dots$ denote contributions from other vertex operators and other descendants, and the OPE coefficient is given by
\begin{equation}
  C(\Phi_1, \Phi_2, w_{1,-2}) = C(\Phi_1, \Phi_2, \V_h) \times Q_{12}
  \times \left(\frac{h-h'}{h-h^{(K)}_{12}}\right) \,.
\end{equation}

Finally, combining the definition~\eqref{eq:partner} of $w_{1,-2}$ and the well-known two-point functions for $\nt \neq n$:
\begin{equation}
  \begin{aligned}
    &\aver{\V_h(z,\zb) \V_h(0)} = |z|^{-4h} \,,
    \qquad
    \aver{\W_{1,-2}(z,\zb) \W_{1,-2}(0)} = |z|^{-4h'} \,, \\
    &\aver{\V_h(z,\zb) \W_{1,-2}(0)} = 0 \,,
  \end{aligned}
\end{equation}
we obtain the two-point functions in the limit $\nt \to n$:
\begin{equation}
  \boxmath{\begin{aligned}
    &\aver{\chi_{12}(z,\zb) \chi_{12}(0)} = 0 \,,
    \qquad
    \aver{w_{1,-2}(z,\zb) w_{1,-2}(0)} =
    \frac{\theta - 2\beta \log(z\zb)}{z^{2h_{1,-2}} \ \zb^{2h_{12}}} \,, \\
    &\aver{\chi_{12}(z,\zb) w_{1,-2}(0)} = \frac{\beta}{z^{2h_{1,-2}} \ \zb^{2h_{12}}} \,,
  \end{aligned}}
\end{equation}
where $\theta$ is arbitrary, and the indecomposability parameter $\beta$ is given by
\begin{equation}
  \beta = \lim_{e_0 \to 1-g} \frac{\aver{\V_h|(L_{2}-gL_{1}^2)(L_{-2}-gL_{-1}^2)|\V_h}}{h-h'} \,,
\end{equation}
where $h=h_{0,e_0+1}$ and $h'=h_{12}$. A little algebra gives the final result
\begin{equation}
  \beta = 4(1-g) \,.
\end{equation}

\section{Conclusion}
\label{sec:concl}

In this work, we have shown how to adapt the conformal bootstrap approach to the case of operators with conformal spin, and applied it to compute some OPE coefficients involving some mixed electric-magnetic operators in the {\On} loop model. It turns out that, in all the cases where this approach applies, the result is particularly simple: the OPE coefficient is the geometric mean of the timelike-Liouville DOZZ formula~\cite{DO92,ZZ96} on the holomorphic and anti-holomorphic sectors. Our results are supported by economical, yet precise numerical calculations based on the transfer matrix of the loop model. In doing so, we used a numerical method which can be exploited more extensively to test many predictions on OPE coefficients.

Many OPE coefficients remain out of reach for this conformal bootstrap approach, including some very natural observables from the geometric point of view, e.g. the probability that three points lie on the same loop. We have traced these difficulties to the non-diagonalisability of the dilatation operator $(L_0+\bar L_0)$, and proposed a simplistic analysis of the representation theory of the periodic Temperley-Lieb algebra, along the lines of \cite{Dubail10,Vasseur12,Azat12,Azat13,Azat14}, to account for this effect at the lattice level. Similarly to the representation theory of quantum groups, this problem is much simpler for generic values of the loop fugacity $n=q+q^{-1}$ when $q$ is not a root of unity, which is the case under study, but many questions still remain unsolved, the most important one being: do the null vectors descending from electric operators $\ket{\V_{1k}}$ become decoupled or not in the scaling limit? Answering this type of questions would require a substantially deeper analysis of the loop model for generic loop fugacity, which we leave for future work.

\section*{Acknowledgements} The authors wish to thank Gesualdo Delfino, Jacopo Viti, Raoul Santachiara, Marco Picco, Jesper Jacobsen, Hubert Saleur and Vladimir Dotsenko for fruitful discussions.

\section*{Appendix}
\appendix

\section{Details of the Coulomb-gas approach}
\label{sec:app-boot}

The normalisation factor appearing in~\eqref{eq:Fk-lim} is:
\begin{align}
  N_k(\alpha_1, \alpha_2, \alpha_3, \alpha_4) =& \prod_{j=0}^{k-2} \frac{\Gamma[(j+1)\rho]}{\Gamma(\rho)}
  \ \frac{\Gamma(1+a+j\rho)\Gamma(1+c+j\rho)}{\Gamma[2+a+c+(k-2+j)\rho]} \nn \\
  & \times \prod_{j=0}^{p-k-1} \frac{\Gamma[(j+1)\rho]}{\Gamma(\rho)}
  \ \frac{\Gamma(1+b+j\rho)\Gamma(1+d+j\rho)}{\Gamma[2+b+d+(p-k-1+j)\rho]} \,.
  \label{eq:Nk}
\end{align}
The matrix elements on the last row and last column of $A$ [see~\eqref{eq:change}] are:
\begin{align}
  A_{pk} =& \prod_{j=0}^{k-2} \frac{s(1+b+j\rho)}{s[2+b+c+(k-2+j)\rho]}
  \prod_{j=0}^{p-k-1} \frac{s(1+d+j\rho)}{s[2+a+d+(p-k-1+j)\rho]} \,, 
  \label{eq:Apk} \\
  A_{kp} =& \prod_{j=0}^{p-2} s[(j+1)\rho]
  \prod_{j=0}^{k-2} \frac{s[1+b+(p-k+j)\rho]}{s[(j+1)\rho]\ s[b+c+(p-2+j)\rho]} \nn \\
  &\times \prod_{j=0}^{p-k-1} \frac{s[1+c+(k-1+j)\rho]}
  {s[(j+1)\rho]\ s[b+c+(p-3+k+j)\rho]} \,,
  \label{eq:Akp}
\end{align}
where $s(x)=\sin(\pi x)$.
The general matrix elements can be found in~\cite{FK89}. It is also useful to notice that
\begin{equation}\label{eq:Akp2}
  A_{kp} \propto \left\{
    \prod_{j=0}^{k-2} s[(j+1)\rho] s(1+c+j\rho)
    \prod_{j=0}^{p-k-1} s[(j+1)\rho] s(1+b+j\rho)
  \right\}^{-1} \,.
\end{equation}

We shall give here the explicit calculation leading from the expansion coefficients~\eqref{eq:Sk2} to the expression~\eqref{eq:C1} of the OPE constants. First, let us notice that the case $k=1$ corresponds to $C^2(\V_\alpha,\V_\alpha,\V_{11})=1$, since $\V_{11}=\id$. If we set $\alpha=\alpha_{1p}$, and denote $a_p=2\alpha_+\alpha_{1p}$ and $a'_p=2\alpha_+(2\alpha_0-\alpha_{1p})$, we get
\begin{align}
  C^2(\V_{1,2k-1},\V_{1p},\V_{1p}) &= \frac{S_k^{(p)}(a_p,a'_p)}{S_1^{(p)}(a_p,a'_p)} \nn \\
  &= \prod_{j=0}^{p-2} \frac{\gamma[(j+1)\rho]}{\gamma[-1+(j+2)\rho]}
  \times \prod_{j=0}^{k-2} \frac{\gamma[(j+1)\rho]\ \gamma[-1+(p+1+j)\rho]}{\gamma[(j+k)\rho]\ \gamma[(p-1-j)\rho]} \nn \\
  & \times\prod_{j=0}^{p-k-1} \frac{\gamma[(j+1)\rho]\ \gamma[-1+(p+k-1-j)\rho]}{\gamma^2[(p-1-j)\rho]} \,.
  \label{eq:Cppk}
\end{align}
Going back to general $\alpha$, we have
\begin{equation}
  C(\V_\alpha, \V_\alpha, \V_{1,2k-1}) C(\V_{1,2k-1},\V_{1p},\V_{1p}) = \frac{S_k^{(p)}(a,a')}{S_1^{(p)}(a,a')} \,,
\end{equation}
which gives, combined with~\eqref{eq:Cppk}, the result~\eqref{eq:C1}.

\begin{figure}
  \begin{center}
    \begin{tabular}{ccc}
      \scalebox{1}{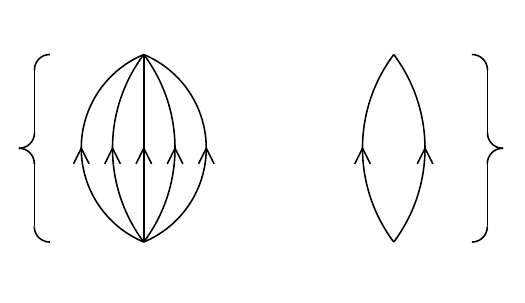}
      &&
      \scalebox{1}{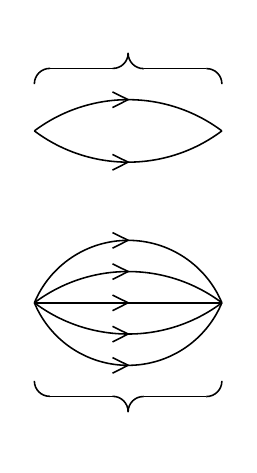} \\
      \\
      (a) &\qquad\qquad\qquad\qquad& (b)
    \end{tabular}
  \end{center}
  \caption{Integration contour defining: (a) the conformal block $\F_k(z)$; (b) the conformal block $\Ft_\ell(z)$.}
  \label{fig:Ck}
\end{figure}

\section{Hypergeometric conformal blocks}
\label{sec:app-hyper}

In this section, we recall the different bases of solutions of the PDE 
\begin{align}
  \left[ g\partial_z^2 + \left( \frac{1}{z} + \frac{1}{z-1}\right)\partial_z + \left(\frac{h_1}{z} - \frac{h_3}{z-1} + h_2- h_4 \right) \frac{1}{z(z-1)} \right] f(z) = 0 \,. \label{holomorphic PDE}
\end{align}
Like in~\secref{funct}, we use a Coulomb gas parametrisation of the conformal dimensions $h_i = \alpha_i (\alpha_i - 2\alpha_0)$, and we introduce
\begin{equation}
  \begin{aligned}
    a = 2\alpha_+ \alpha_1  \,, \qquad  & b = 2\alpha_+ \alpha_3 \,,  \qquad d = 2\alpha_+ \alpha_4 \,, \\
    a'= 2\alpha_+ \alpha'_1 \,, \qquad  & b'= 2\alpha_+ \alpha'_3 \,, \qquad d'= 2\alpha_+ \alpha'_4  \,.
  \end{aligned}
\end{equation}

\subsection{Generic case :  $a-a'$ and $b-b' \notin 2\Zb$}

Several bases of conformal blocks can be considered. In~\secref{funct}, we refer to two such bases. The first one $\{I_1(z), I_2(z)\}$ has simple monodromy as the variable $z$ loops around $0$ and corresponds to the expansion around the singularity $z=0$ : 
\begin{align}
  I_1(z) & =  z^{-\frac{a}{2}}(1-z)^{-\frac{b}{2}}\ {}_2F_{1}\left(
    {\textstyle\frac{d-a-b+\rho}{2}, \frac{d'-a-b+\rho}{2}; 1 + \frac{a'-a}{2}; z}
  \right) \,, \label{I} \\ 
  I_2(z) & = z^{-\frac{a'}{2}}(1-z)^{-\frac{b'}{2}}\ {}_2F_{1}\left(
    {\textstyle\frac{d'-a'-b'+\rho}{2}, \frac{d-a'-b'+\rho}{2}; 1 + \frac{a - a'}{2}; z}
  \right) \,, \label{I'}
\end{align}
while the second basis $\{J_1(z), J_2(z)\}$ has simple monodromy around $z=1$:
\begin{align}
  J_1(z) &=  z^{-\frac{a}{2}}(1-z)^{-\frac{b}{2}}\ {}_2F_{1}\left(
    {\textstyle\frac{d-a-b+\rho}{2}, \frac{d'-a-b+\rho}{2}; 1 + \frac{b'-b}{2}; 1-z}
  \right) \,,  \label{J} \\ 
  J_2(z)  &= z^{-\frac{a'}{2}}(1-z)^{-\frac{b'}{2}}\ {}_2F_{1}\left(
    {\textstyle\frac{d'-a'-b'+\rho}{2}, \frac{d-a'-b'+\rho}{2}; 1 + \frac{b-b'}{2}; 1-z}
  \right) \,.  \label{J'} 
\end{align}
It is useful to know how these conformal blocks transform under the various re\-pa\-ra\-me\-teri\-sations $\alpha_i \leftrightarrow \alpha_i'$:
\begin{itemize}
\item Under the exchange $a \leftrightarrow a'$, we have $I_1(z) \leftrightarrow I_2(z)$, while $J_1(z)$ and $J_2(z)$ are left invariant\footnotemark\footnotetext{To see this, the relation ${}_2F_1 (a,b;c;z) =  (1-z)^{c-a-b}{}_2F_1 (c-a, c-b;c ; z)$ comes in handy.}.  
\item Likewise, under $b \leftrightarrow b'$, we have $J_1(z) \leftrightarrow J_2(z)$, $I_1$ and $I_2(z)$ being left invariant.  
\item Finally, under $d \leftrightarrow d'$, the conformal blocks $\{I_1(z), I_2(z) \}$ and $\{J_1(z), J_2(z)\}$ are invariant. 
\end{itemize}
The change of basis is given by the matrix $M(a,b,d)$: 
\begin{equation}
  I_k(z) = \sum_{\ell=1}^2 M_{k\ell}(a,b,d) \ J_\ell(z) \,,
\end{equation}
where
\begin{align}
  M(a,b,d) = \left( \begin{array}{cc}
      \frac{\Gamma\left(1 + \frac{a' - a}{2}\right) \Gamma\left(\frac{b - b'}{2}\right)}{\Gamma\left( \frac{d-a-b'+\rho}{2} \right) \Gamma\left( \frac{d'-a-b'+\rho}{2} \right)} & 
      \frac{\Gamma\left(1 + \frac{a' - a}{2}\right) \Gamma\left(\frac{b' - b}{2}\right)}{\Gamma\left( \frac{d-a-b+\rho}{2} \right) \Gamma\left( \frac{d'-a-b+\rho}{2} \right)}     \\  
      \frac{\Gamma\left(1 + \frac{a - a'}{2}\right) \Gamma\left(\frac{b - b'}{2}\right)}{\Gamma\left( \frac{d-a'-b'+\rho}{2} \right) \Gamma\left( \frac{d'-a'-b'+\rho}{2} \right)}   &
      \frac{\Gamma\left(1 + \frac{a - a'}{2}\right) \Gamma\left(\frac{b' - b}{2}\right)}{\Gamma\left( \frac{d-a'-b+\rho}{2} \right) \Gamma\left( \frac{d'-a'-b+\rho}{2} \right)} \end{array} \right) \,. \label{change of basis M}
\end{align}
As long as all four conformal blocks are well defined (i.e. $a'-a \notin 2\Zb$ and $b'-b \notin 2\Zb$), this matrix is invertible, with inverse $M^{-1}(a,b,d) = M(b,a,d)$. 

\subsection{Logarithmic case}

Suppose now the indicial equation around the singularity $z=0$ is degenerate, namely $a-a' \in 2\Zb$. Without any loss of generality we can assume $a'-a =2 m $ where $m$ is a positive integer.  Since we have $a'+ a = 2\rho -2$, this means
\begin{align}
  a = \rho -1 - m, \qquad a' = \rho -1 + m \,.
\end{align}
In that case, the conformal block given by \eqref{I} is still a solution,
\begin{align}
  I_1(z) = z^{\frac{1+m-\rho}{2}}(1-z)^{-\frac{b}{2}}{}_2F_{1}\left(
    {\textstyle \frac{d-b+m+1}{2}, \frac{d'-b+m+1}{2}; 1 + m; z}
  \right) \,.
\end{align}
On the other hand, $\eqref{I'}$ is no longer well-defined, since its third argument is a non-positive integer. We can approach the degenerate case as
\begin{align}
  a = \rho -1 - m - \epsilon, \qquad a' = \rho -1 + m + \epsilon \,,
  \label{regularized a a'}
\end{align} 
with small $\epsilon$. Let us denote $\{ \Ie(z) , \Ipe(z) \}$ the corresponding solutions. The first conformal block encounters no issues as $\epsilon \to 0$ : 
\begin{align}
  \lim_{\epsilon \to 0} \Ie(z)& = I_1(z) \,.
\end{align}
On the other hand, the second conformal block, as given by $\eqref{I'}$, has a first order pole 
\begin{align}
  \Ipe(z)  & \sim \frac{\Gamma(1-m-\epsilon)}{m!}\left(\frac{d-b-m+1}{2}\right)_m\left(\frac{d'-b-m+1}{2}\right)_m \, I_1(z) \,,
\end{align}
where $(x)_n = x(x+1) \dots (x+n-1)$ is the Pochhammer symbol.
For instance the second conformal block can be chosen as the limit 
\begin{align}
  \wt I_2(z) =  \frac{(-1)^{m} m!}{2} \lim_{\epsilon \to 0} \Bigg[
  & \frac{\Gamma \left(\frac{d-a'-b+\rho}{2}\right)\Gamma \left(\frac{d'-a'-b+\rho}{2}\right)}{ \Gamma \left( 1 + \frac{b'-b}{2}\right) } \Ie(z) \nn \\
  & + \frac{\Gamma \left(\frac{d-a'-b'+\rho}{2}\right)\Gamma \left(\frac{d'-a'-b'+\rho}{2}\right)}{ \Gamma \left( 1 + \frac{b-b'}{2}\right) } \Ipe(z)
  \Bigg] \,,
\end{align}
with $a$ and $a'$ given by \eqref{regularized a a'}. Note that  this expression still makes sense  for $b'-b \in 2\Zb$, since the r.h.s. has a well-defined limit as $b'-b  \to 2n$ , $n \in \Zb$. Moreover we have chosen the second conformal block to be invariant under $b \leftrightarrow b'$. This gives a solution of the form
\begin{align}
  \wt I_2(z) =  I_1(z)\,\log z + H(z)\,z^{(1-m-\rho)/2}  \,,
\end{align}
where $H(z)$ is a regular function around $z=0$, whose explicit expression we will not need. The monodromy around $z=0$ is no longer diagonal:
\begin{equation}
  I_1(z) \to e^{-i\pi a}\ I_1(z) \,,
  \qquad \wt I_2(z) \to e^{-i\pi a} \left[ 2i\pi I_1(z) + \wt I_2(z) \right] \,.
\end{equation}
If $b'-b \notin 2\Zb$, the basis $\{J_1(z), J_2(z) \}$ is the same as for the generic case \eqref{J}-\eqref{J'}, and the change of basis becomes
\begin{align*}
  \wt{M}(a,b,d) =  m ! \left( \begin{array}{cc}
      \frac{ \Gamma\left(\frac{b - b'}{2}\right)}{\Gamma\left( \frac{d-b'+1+m}{2} \right) \Gamma\left( \frac{d'-b'+1+m}{2} \right)} & 
      \frac{ \Gamma\left(\frac{b' - b}{2}\right)}{\Gamma\left( \frac{d-b+1+m}{2} \right) \Gamma\left( \frac{d'-b+1+m}{2} \right)}     \\  
      \frac{(-1)^{m} \Gamma\left( \frac{d-b+1-m}{2} \right) \Gamma\left( \frac{d'-b+1-m}{2} \right)}{ 2 \Gamma\left(1+ \frac{b' - b}{2}\right)}   &
      \frac{(-1)^{m}  \Gamma\left( \frac{d-b'+1-m}{2} \right) \Gamma\left( \frac{d'-b'+1-m}{2} \right)}{ 2 \Gamma\left(1+ \frac{b - b'}{2}\right)}
    \end{array} \right) \,,
\end{align*}
\begin{align*}
\wt{M}^{-1}(a,b,d) = \frac{1}{m!}  \left( \begin{array}{cc}
 \frac{\Gamma\left( \frac{d-b'+1+m}{2} \right) \Gamma\left( \frac{d'-b'+1+m}{2} \right)}{2 \Gamma\left(\frac{b - b'}{2}\right)}  & 
(-1)^{m} \frac{  \Gamma\left(1+ \frac{b' - b}{2}\right)} {\Gamma\left( \frac{d-b+1-m}{2} \right) \Gamma\left( \frac{d'-b+1-m}{2} \right)}    \\  
  \frac{\Gamma\left( \frac{d-b+1+m}{2} \right) \Gamma\left( \frac{d'-b+1+m}{2} \right)}{2 \Gamma\left(\frac{b' - b}{2}\right)}      &
(-1)^{m}  \frac{  \Gamma\left(1+ \frac{b - b'}{2}\right)} {\Gamma\left( \frac{d-b'+1-m}{2} \right) \Gamma\left( \frac{d'-b'+1-m}{2} \right)}  \end{array} \right) \,.
\end{align*}

\section{Constraints on the spectrum from locality}
\label{sec:app-locality}

In this appendix we explain how locality imposes the constraint :
\begin{align}
  a-\bar{a} \in \Zb \qquad \text{or} \qquad a'-\bar{a} \in \Zb \,.
\end{align}
We first prove this in the generic case ($a'-a \notin 2\Zb$) before doing the $\log$ case  ($a'-a \in 2\Zb$).

\subsection{Generic case}

Let us consider the four-point function
\begin{align}
  G(z,\bar{z}) = \langle \Phi_{h_1,\hb_1}(0)   \Phi_{1,2}(z) \Phi_{1,2}(1)\Phi_{h_1,\hb_1}(\infty) \rangle \,.
\end{align}
Note that this four-point function cannot vanish due to the trivial fusion channel to the vacuum. To describe this four-point function we can choose $b = \bb = 3\rho-2$ and $d = a$. If there are two fusion channels, then we must have $a - \bar{a} \in \Zb$ or $a' - \bar{a} \in \Zb$ as we explained in the main text. Suppose now there is a single fusion channel around $z=0$. Using the freedom to change $a \to a'$, we can assume 
\begin{equation}
  G(z,\bar{z}) =  X_1 I_1(z) \bar I_1(\zb) \,.
\end{equation}
Since there is a single fusion channel as $z\to0$, there is also a single fusion channel as $z\to 1$\footnotemark \footnotetext{From the matrix $X$ we get the expression of the matrix $Y = M^t X \bar{M}$. These two matrices share the same rank, and the number of fusion channels around  $z=0$ and $z=1$ has to be the same.}. There we know that this fusion has to be $ \Phi_{1,2} \times \Phi_{1,2} \to 1$, which corresponds to 
\begin{equation}
  G(z,\bar{z}) =  J_1(z) \bar J_1(\zb) \,.
\end{equation}
Consistency of these two relations enforces $M_{12} = \bar{M}_{12} =0$, where $M$ is given by \eqref{change of basis M} with $b = 3\rho-2$ and $d = a$. This requires respectively $a \in \Zb$ and $\ab \in \Zb$, and in particular $a - \ab \in \Zb$ is obeyed. 

\subsection{Logarithmic case}

Let us now consider the logarithmic case, namely $a'-a = 2m$ with  $m \in \Nb$. This means $a = \rho -1 - m$ and $a' = \rho -1 + m$. If we also have $\ab'-\ab = 2\bar{m}$ with  $\bar{m} \in \Nb$, then $a-\bar{a} = \bar{m}-m \in \Zb$. Out of interest we mention the form of the four-point function compatible with locality:
\begin{align}
G(z,\bar{z}) = X_1 \, I_1(z) \bar I_1(\zb) + X_2 \, \left[I_1(z) \overline{\wt I_2}(\zb)+ \wt I_2(z) \bar I_1(\zb)\right] \,,
\end{align}
with monodromy factor $e^{i\pi(\bar{a} -a)}$ as in the generic case. 

This only leaves the hybrid case to consider : a logarithmic holomorphic side, say $a'-a = 2m$ with  $m \in \Nb$,  coupled to a generic antiholomorphic behavior, $\ab'-\ab \notin 2\Zb$. In this case locality demands a single fusion channel, which we can choose to be 
\begin{align}
G(z,\bar{z}) = X_1 \,  I_1(z) \bar I_1 (\zb) \,.
\end{align}
The same arguments as in the generic case apply, and we find that $a$ and $\ab$ have to be integers. This is impossible since $a = \rho -1 - m$. Therefore the hybrid case is ruled out by locality. 

\section{Numerical determination of OPE coefficients from the transfer matrix}
\label{sec:app-num}

In this section, we describe a simple algorithm to compute OPE coefficients numerically using the eigenstates of the transfer matrix or Hamiltonian. Only the principle of the method is given here, and we will expose all the implementation details for its application to loop models in a future publication.

Let $(\phi_1, \phi_2, \phi_3)$ be three quasi-primary operators of a CFT, and suppose we want to determine numerically the OPE coefficient $C(\phi_1, \phi_2, \phi_3)$. This coefficient can be expressed as the scalar product
\begin{equation}
  C(\phi_1, \phi_2, \phi_3) = \aver{\phi_1|\phi_2|\phi_3} \,.
\end{equation}
Let us consider the transfer matrix $t_L$ for a periodic system of circumference $L$, denote by $\ket{\phi_j}_L$ the eigenstate of $t_L$ corresponding to the state $\ket{\phi_j}$ in the CFT, and suppose that $\ket{\phi_1}_L$ and $\ket{\phi_3}_L$ are obtained numerically by a partial diagonalisation of $t_L$. Moreover, suppose that one is able to construct a lattice operator $\phi_2^{(L)}(\vec r)$ which behaves as $\phi_2^{(L)} \sim N_2 L^{-x_2} \ \phi_2$ in the scaling limit, where $N_2$ is some non-universal normalisation constant. More precisely, there will exist a series of operators $\phi_2, \phi'_2 \dots$ of increasing scaling dimensions $x_2<x'_2<\dots$ and some normalisation factors $N_2, N'_2, \dots$ so that 
\begin{equation} \label{eq:phi2}
  \phi_2^{(L)} = N_2 L^{-x_2} \ \phi_2 + N'_2 L^{-x_2'} \ \phi'_2 + \dots
\end{equation}
in the scaling limit. Clearly, with the standard normalisation of CFT states $\aver{\phi_j|\phi_j}=1$, and denoting $\ket 0$ the ground state, one has $\phi_2^{(L)}\ket{0}_L \sim N_2 L^{-x_2} \ket{\phi_2}$. Then one obtains:
\begin{equation} \label{eq:C-num}
  \frac{_L\aver{\phi_1| \phi_2^{(L)}| \phi_3}_L}{_L\aver{\phi_2| \phi_2^{(L)}| 0}_L} = C(\phi_1, \phi_2, \phi_3)
  + O(1/L^{x''_2-x_2}) \,,
\end{equation}
where $x''_2$ is the scaling dimension of the most relevant subleading operator in the expansion~\eqref{eq:phi2} contributing to $_L\aver{\phi_1| \phi_2^{(L)}| \phi_3}_L$ or $_L\aver{\phi_2| \phi_2^{(L)}| 0}_L$.

\bibliographystyle{unsrt}
\bibliography{biblio}

\begin{thebibliography}{10}

\bibitem{DSZ87}
Ph~di~Francesco, H~Saleur, and J-B Zuber.
\newblock {Relations between the Coulomb gas picture and conformal invariance
  of two-dimensional critical models}.
\newblock {\em J. Stat. Phys.}, 49:57--79, 1987.

\bibitem{DV11}
G~Delfino and J~Viti.
\newblock {On three-point connectivity in two-dimensional percolation}.
\newblock {\em J. Phys. A: Math. Theor.}, 44:032001, 2011.

\bibitem{PSVD13}
M~Picco, R~Santachiara, J~Viti, and G~Delfino.
\newblock {Connectivities of Potts Fortuin-Kasteleyn clusters and time-like
  Liouville correlator}.
\newblock {\em Nucl. Phys. B}, 875:719--737, 2013.

\bibitem{Teschner95}
J~Teschner.
\newblock {On the Liouville three point function}.
\newblock {\em Phys. Lett. B}, 363:65--70, 1995.

\bibitem{Saleur87}
H~Saleur.
\newblock {Conformal invariance for polymers and percolation}.
\newblock {\em J. Phys. A}, 20:455--470, 1987.

\bibitem{RozSaleur91}
L~Rozansky and H~Saleur.
\newblock {Quantum field theory for the multivariable Alexander Conway
  polynomial}.
\newblock {\em Nucl. Phys. B}, 376:461--509, 1991.

\bibitem{Gurarie93}
V~Gurarie.
\newblock {Logarithmic operators in conformal field theory}.
\newblock {\em Nucl. Phys. B}, 410:535--549, 1993.

\bibitem{Pearce06}
P~Pearce, J~Rasmussen, and J-B Zuber.
\newblock {Logarithmic Minimal Models}.
\newblock {\em J. Stat. Mech.}, 0611:017, 2006.

\bibitem{Dubail10}
J.L.~Jacobsen J.~Dubail and H.~Saleur.
\newblock {Conformal field theory at central charge $c=0$: a measure of the
  indecomposability ($b$) parameters}.
\newblock {\em Nucl. Phys. B}, 834:399--422, 2010.

\bibitem{Vasseur12}
J.L.~Jacobsen R.~Vasseur, A.M.~Gainutdinov and H.~Saleur.
\newblock {Puzzle of bulk conformal field theories at central charge $c=0$}.
\newblock {\em Phys. Rev. Lett.}, 108:161602, 2012.

\bibitem{Azat12}
A~Gainutdinov, J~L Jacobsen, H~Saleur, and R~Vasseur.
\newblock {A physical approach to the classification of indecomposable Virasoro
  representations from the blob algebra}.
\newblock {\em Nucl. Phys. B}, 873:614--681, 2013.

\bibitem{Azat13}
A~Gainutdinov, J~L Jacobsen, N~Read, H~Saleur, and R~Vasseur.
\newblock {Logarithmic conformal field theory: a lattice approach}.
\newblock {\em J. Phys. A: Math. Theor.}, 46,:494012, 2013.

\bibitem{Azat14}
A~M Gainutdinov, N~Read, H~Saleur, and R~Vasseur.
\newblock {The periodic sl($2|1$) alternating spin chain and its continuum
  limit as a bulk Logarithmic Conformal Field Theory at $c=0$}.
\newblock {\em arXiv:1409.0167}, 2014.

\bibitem{Nienhuis82}
B~Nienhuis.
\newblock {Exact critical point and critical exponents of O($n$) models in two
  dimensions}.
\newblock {\em Phys. Rev. Lett.}, 49:1062--1065, 1982.

\bibitem{DF84}
Vl~S Dotsenko and V~A Fateev.
\newblock {Conformal algebra and multipoint correlation functions in 2d
  statistical models}.
\newblock {\em Nucl. Phys. B}, 240:312--348, 1984.

\bibitem{DF85}
Vl~S Dotsenko and V~A Fateev.
\newblock {Four-point correlation functions and the operator algebra in 2d
  conformal invariant theories with central charge $c \leq 1$}.
\newblock {\em Nucl. Phys. B}, 251:691--734, 1985.

\bibitem{ZZ96}
A~B Zamolodchikov and Al~B Zamolodchikov.
\newblock {Conformal bootstrap in Liouville field theory}.
\newblock {\em Nucl. Phys. B}, 477:577--605, 1996.

\bibitem{Baxter-book}
R~J Baxter.
\newblock {\em {Exactly Solved Models in Statistical Mechanics}}.
\newblock {Dover Publications (Mineola, New York)}, 1982.

\bibitem{DO92}
H~Dorn and H~J Otto.
\newblock {On correlation functions for non-critical strings with $c \leq 1$
  but $d \geq 1$}.
\newblock {\em Phys. Lett. B}, 291:39--43, 1992.

\bibitem{DO94}
H~Dorn and H~J Otto.
\newblock {Two- and three-point functions in Liouville theory}.
\newblock {\em Nucl. Phys. B}, 429:375--388, 1994.

\bibitem{IC09}
Y~Ikhlef and J~Cardy.
\newblock {Discretely Holomorphic Parafermions and Integrable Loop Models}.
\newblock {\em J. Phys. A: Math. Theor.}, 42:102001, 2009.

\bibitem{IWWZ}
M.~Wheeler Y.~Ikhlef, R.~Weston and P.~Zinn-Justin.
\newblock {Discrete holomorphicity and quantized affine algebras}.
\newblock {\em J. Phys. A: Math. Theor.}, 46:265205, 2013.

\bibitem{Cardy01}
J~Cardy.
\newblock {The stress tensor in quenched random systems}.
\newblock {Talk presented at Workshop on Statistical Field Theory, Como, Italy,
  June 2001, cond-mat/0111031v1}.

\bibitem{FK89}
J~Fuchs and A~Klemm.
\newblock {The computation of the operator algebra in non-diagonal conformal
  field theories}.
\newblock {\em Annals of Physics}, 194:303--335, 1989.

\end{thebibliography}

\end{document}